\documentclass[prodmode,acmtecs]{acmsmall}

\usepackage[ruled]{algorithm2e}

\SetAlFnt{\small}
\SetAlCapFnt{\small}
\SetAlCapNameFnt{\small}
\SetAlCapHSkip{0pt}
\IncMargin{-\parindent}

\newcommand{\pdffigure}[3]{
\begin{figure}
\centering
\includegraphics[width=120mm]{#2}  
\caption{#3} \label{#1}
\end{figure}
}

\usepackage{amsfonts}
\usepackage{amssymb}
\usepackage{graphicx}
\usepackage{enumerate}

\newfont{\Fr}{eufm10 scaled\magstep1}
\newfont{\Sc}{eusm10 scaled\magstep1}
\newfont{\Bb}{msbm10 scaled\magstep1}
\newfont{\Er}{eurm10 scaled\magstep1}
\newfont{\Msa}{msam10 scaled\magstep1}

\newenvironment{remarks}{\medskip\noindent{\it Remarks.}\begin{enumerate}}{\end{enumerate} \medskip}
\newenvironment{note}{\medskip\noindent{\it Note.}}{\medskip}

\newcommand{\beeq}[1]{\begin{equation} \label{#1}}
\newcommand{\eeq}{\end{equation}}
\newcommand{\beeqs}{\begin{eqnarray*}}
\newcommand{\eeqs}{\end{eqnarray*}}
\renewcommand{\(}{\begin{eqnarray*}}
\renewcommand{\)}{\end{eqnarray*}}
\newcommand{\beeqn}{\begin{eqnarray}}
\newcommand{\eeqn}{\end{eqnarray}}
\newcommand{\sothat}{\\ \Rightarrow~~~~ &&}
\newcommand{\nexteqline}{\\ &=&}
\newcommand{\eqand}{\mbox{~~~and~~~}}
\newcommand{\eqwhere}{\mbox{~~~where~~~}}

\newcommand{\refth}[1]{Theorem~\ref{#1}}
\newcommand{\reflm}[1]{Lemma~\ref{#1}}
\newcommand{\refco}[1]{Corollary~\ref{#1}}
\newcommand{\refprop}[1]{Proposition~\ref{#1}}
\newcommand{\refsec}[1]{Section~\ref{#1}}
\newcommand{\reffig}[1]{Fig.~\ref{#1}}
\newcommand{\refeq}[1]{(\ref{#1})}

\renewcommand{\quad}{\hspace*{3mm}}
\renewcommand{\qquad}{\hspace*{5mm}}

\newcommand{\lp}{\left(  }
\newcommand{\rp}{\right) }
\newcommand{\lb}{\left\{  }
\newcommand{\rb}{\right\} }
\newcommand{\lbr}{\left[  }
\newcommand{\rbr}{\right] }
\newcommand{\lf}{\left\lfloor}
\newcommand{\rf}{\right\rfloor}

\newcounter{cnt1}
\newcounter{cnt2}
\newcounter{cnt3}
\newcounter{cnt4}
\newcounter{cnt5}
\newcounter{cnta}

\newcommand{\beenu}
{
\begin{list}{\arabic{cnt1}.}
{\usecounter{cnt1}
\leftmargin 7mm
\setlength{\leftmargin}{\leftmargin}
\topsep 1pt
\parsep 1pt
\itemsep 1pt}
}
\newcommand{\eenu}{\end{list}}

\newcommand{\beenub}
{
\begin{list}{\arabic{cnt1}-\arabic{cnt2}.}
{
\leftmargin  7mm
\setlength{\leftmargin}{\leftmargin}
\topsep 1pt
\parsep 1pt
\itemsep 1pt}
}
\newcommand{\eenub}{\end{list}}
\newcommand{\itemb}{\addtocounter{cnt2}{1} \item}
\newcommand{\itembb}{\setcounter{cnt2}{1} \item}

\newcommand{\beenuc}
{
\begin{list}{\arabic{cnt1}-\arabic{cnt2}-\arabic{cnt3}.}
{
\leftmargin  2mm
\setlength{\leftmargin}{\leftmargin}
\topsep 1pt
\parsep 1pt
\itemsep 1pt}
}
\newcommand{\eenuc}{\end{list}}
\newcommand{\itemc}{\addtocounter{cnt3}{1} \item}
\newcommand{\itemcc}{\setcounter{cnt3}{1} \item}

\newcommand{\beenud}
{
\begin{list}{\arabic{cnt1}-\arabic{cnt2}-\arabic{cnt3}-\arabic{cnt4}.}
{
\leftmargin  2mm
\setlength{\leftmargin}{\leftmargin}
\topsep 1pt
\parsep 1pt
\itemsep 1pt}
}
\newcommand{\eenud}{\end{list}}

\newcommand{\beenue}
{
\begin{list}{\arabic{cnt1}-\arabic{cnt2}-\arabic{cnt3}-\arabic{cnt4}-\arabic{cnt5}.}
{
\leftmargin  5mm
\setlength{\leftmargin}{\leftmargin}
\topsep 1pt
\parsep 1pt
\itemsep 1pt}
}
\newcommand{\eenue}{\end{list}}

\newcommand{\beitm}
{
\begin{list}{$\bullet$}
{
\leftmargin  6mm
\setlength{\leftmargin}{\leftmargin}
\topsep 0pt
\parsep 0pt
\itemsep 0pt}
}
\newcommand{\eitm}{\end{list}}

\newcommand{\beitma}
{
\begin{list}{(\alph{cnta})~~}
{\usecounter{cnta}
\labelsep 0mm
\leftmargin  10mm
\setlength{\leftmargin}{\leftmargin}
\topsep 0pt
\parsep 0pt
\itemsep 0pt}
}
\newcommand{\eitma}{\end{list}}

\newcommand{\bbegin}{{\bf begin~}}
\newcommand{\bend}{{\bf end~}}
\newcommand{\bif}{{\bf if~}}
\newcommand{\bthen}{{\bf then~}}
\newcommand{\belse}{{\bf else~}}
\newcommand{\bfor}{{\bf for~}}
\newcommand{\bto}{{\bf to~}}
\newcommand{\bdo}{{\bf do~}}

\newcommand{\breturn}{{\bf return~}}

\newcommand{\spca}{{\hspace*{3mm}}}
\newcommand{\spcb}{{\hspace*{6mm}}}

\newcommand{\spcd}{{\hspace*{12mm}}}

\newcommand{\bealg}
{
\begin{list}{}
{
\leftmargin 15mm
\rightmargin 15mm
\setlength{\leftmargin}{\leftmargin}
\topsep 20pt
\parsep 0pt
\itemsep 0pt}
\item
\begin{footnotesize}
}
\newcommand{\ealg}{\end{footnotesize} \end{list}}

\newcommand{\Z}{\mbox{\Bb Z}}

\newcommand{\R}{\mbox{\Bb R}}

\newcommand{\PP}{\mbox{\rm{\bf P}}}
\newcommand{\NP}{\mbox{\rm{\bf NP}}}

\newcommand{\ep}{\epsilon}

\begin{document}

\markboth{Fukuyama.}{Exponential Monotone Complexity of the Clique Function}

\title{An Alternative Proof of the Exponential Monotone Complexity of the Clique Function}
\author{Junichiro Fukuyama
\affil{Applied Research Laboratory, The Pennsylvania State University}
}

\begin{abstract}

In 1985, Razborov discovered a proof that the monotone circuit complexity of the clique problem is super-polynomial. Alon and Boppana improved the result into exponential lower bound $exp \lp \Omega \lp \lp n / \log n\rp^{1/3} \rp \rp$ of a monotone circuit $C$ to compute cliques of size $\frac{1}{4} \lp n / \log n \rp^{2/3}$, where $n$ is the number of vertices in a graph. Both proofs are based on the {\em method of approximations} and Erd\"os and Rado's {\em sunflower lemma}. There has been an interest in further generalization of the proof scheme.

In this paper, we present a new approach to show the exponential monotone complexity. Unlike the standard method, it dynamically constructs a counter example: Assuming a monotone circuit $C$ of sub-exponential size to compute $k$-cliques $c$, an algorithm finds an edge set $t$ containing no $c$ in the disjunctive normal form constructed at the root of $C$. We call such $t$ a {\em shift}. The proof shows that $t$ is disjoint from an edge set $z$ whose removal leaves no $k$-cliques.

We explore the {\em set theoretical nature} of computation by Boolean circuits. We develop a theory by finding topological properties of {\em the Hamming space} $2^{[n]}$ where $[n]=\lb 1, 2, \ldots, n \rb$. A structural theorem is presented, which is closely related to the sunflower lemma and claims a stronger statement in most cases. The theory lays the foundation of the above shift method. It also shows the existence of a sunflower with small core in a family of sets, which is not an obvious consequence of the sunflower lemma.

Lastly, we point out that the new methodology has potential to apply to a general circuit computing cliques due to the dynamic selection of $t$ and $z$, and to improve the Alon-Boppana bound $exp \lp \Omega \lp \lp n / \log n\rp^{1/3} \rp \rp$.
\end{abstract}






\maketitle


\section{Introduction}

The monotone circuit complexity of the $k$-clique problem, or {\em monotone complexity of the clique function} {\sc Clique}$_{n, k}$, is the minimum size of a Boolean circuit with no logical negation to compute the $k$-clique problem over $n$ vertices. It was first proven super-polynomial by Razbrov [1985a], and later improved by Alon and Boppana [1987] into an exponential lower bound. The importance of the monotone complexity is that it implies $\PP \ne \NP$ if we can generalize the claim for an arbitrary circuit.

The proof is based on the {\em method of approximations} developed by Razborov, and Erd\"os and Rado's {\em sunflower lemma} \cite{raz3} \cite{papa}. It has been the standard proof technique for intractability of the clique problem in the monotone case. Razborov also proved with a similar scheme [1985b] that the bipartite perfect matching problem has monotone complexity $exp \lp \Omega \lp \log^2 n\rp \rp$. Tardos noted in [1987] that there exists a problem in $\PP$ whose monotone circuit complexity is exponential, $i.e.$, the gap between the monotone and general complexity can be very large for some problems. Alon and Boppana improved Razborov's proof to show the monotone bound $exp \lp \Omega \lp \lp n / \log n\rp^{1/3} \rp \rp$ for $k$-cliques where $k=\frac{1}{4} \lp n / \log n \rp^{2/3}$, demonstrating that the monotone complexity is actually exponential in $n$. Raz and Wigderson showed depth $\Omega \lp n \rp$ of a monotone circuit computing perfect matching or cliques [1987]. In \cite{r08}, it is shown that a possibly non-monotone circuit for $k$-cliques with constant depth has size $\Omega \lp n^{k/4} \rp$.

The method of approximations considers some particular truth assignments to a given circuit $C$ to argue that after some modification, $C$ correctly returns true/false with high probability. The proof technique is generalized into a broader framework called {\em natural proofs}. A natural proof is characterized by its i) {\em constructiveness}: inductive confirmation of a logical property ${\cal P}$ along the node structure of $C$, and ii) {\em largeness}: statistical significance ${\cal P}$ possesses to distinguish $C$'s computational ability from the pseudo-random case. In \cite{raz3}, strong evidence is posed that a natural proof cannot prove a super-polynomial size lower bound of $C$. It means that if we devise a proof of computational hardness for a general circuit $C$, it most probably does not rely on both the constructiveness and largeness.

\medskip

In this paper, we present a new methodology to show the exponential monotone circuit complexity of computing cliques. We explore the {\em set theoretical nature} of computation by Boolean circuits. This could add a new viewpoint to the above investigations on the $k$-clique problem carried out so far. We develop a theory on topological properties of {\em the Hamming space} $2^{[n]}$ where $[n]=\lb 1, 2, \ldots, n \rb$. Our investigation creates general tools on the {\em $l$-extension of a family of $m$-sets} as in \cite{my6}, and its {\em generator}. We will prove the existence of a small generator that produces a vast majority of $l$-sets in a subspace of the Hamming space. It is a structural theorem closely related to the sunflower lemma, which claims a stronger statement in most cases.

Based on the developed theory, we show:
\begin{enumerate} [(a)]
\item the exponential monotone complexity of {\sc Clique}$_{n, \sqrt[4] n}$ with the structural theorem, and
\item the existence of a sunflower with small core in certain families of sets, which is not an obvious consequence of the sunflower lemma.
\end{enumerate}

Unlike the standard method for (a), the proof dynamically constructs a counter example: Assuming a monotone circuit $C$ of sub-exponential size to compute $k$-cliques $c$, an algorithm finds an edge set $t$ containing no $c$ in the disjunctive normal form constructed at the root of $C$. We call such $t$ a {\em shift}. It will be shown that $t$ is disjoint from an edge set $z$ whose removal leaves no cliques of size $k$.

Lastly, we point out that the new methodology has potential to apply to a general circuit to compute cliques due to dynamic selection of $t$ and $z$, and to improve the Alon-Boppana bound $exp \lp \Omega \lp \lp n / \log n\rp^{1/3} \rp \rp$.

The rest of the paper consists as follows: In Section 2, we develop basic tools in the Hamming space. Section 3 shows the existence of a small generator, which is the aforementioned structural theorem related to the sunflower lemma. Section 4 presents the new approach to prove the monotone complexity. It is followed by discussions on the new method in Section 5 and conclusions in Section 6. 

\section{The Hamming Space} \label{HammingSpace}

\subsection{Family of $m$-Sets And Its Extension} \label{family}

Let $[n]=\lb 1, 2, \ldots, n \rb$ for a positive integer $n$. The family $2^{[n]}$ of all subsets of $[n]$ is called the {\em Hamming Space}, or {\em universal space}, whose metric is the Hamming distance. We also say that a subset $X \subseteq [n]$ is a {\em (sub)space of size $|X|$}, treating it as $2^X$. The size $n$ of the universal space is a special positive integer, and we use it throughout the paper without defining it.

For $X\subseteq [n]$ and $m \in [n]$, we denote by ${X \choose m}$ the family of all subsets of $X$ whose cardinality is $m$. Such a subset is called an $m$-set\footnote{In this paper, a parameter before an object denotes the number of elements unless defined otherwise. Also $A\subset B$ means that $A$ is a proper subset/sub-family of $B$ so $A \ne B$.}. The letter $U$ denotes a given family of $m$-sets in the universal space, $i.e.$, $U \subseteq {[n] \choose m}$, whose {\em size} is its cardinality $|U|$. The {\em sparsity of $U$} is
\[
\kappa \lp U \rp = \ln {n \choose m} -   \ln |U|, 
\]
$i.e.$, $|U|={n \choose m}e^{-\kappa \lp U \rp}$. We occasionally emphasize that it is defined in a space $X$. For example, if $X = \lbr \lf n^{1/3} \rf \rbr \subset [n]$ and $U \subseteq {X \choose m}$ such that $|U|= {|X| \choose m} / 2$, the sparsity $\kappa \lp U \rp$ is $\ln 2$ in the space $X$, but is much larger in the universal space in a general case. 

We also say that the family ${[n] \choose 2}$ of 2-sets is the {\em edge space over $[n]$}: the set of all possible edges for a graph with vertex set $[n]$. The family ${{[n] \choose 2} \choose m}$ consists of all the possible edge sets of size $m$.

The {\em complement $\overline{U}$ of $U$} is ${[n] \choose m} \setminus U$ in the universal space. We say
\[
\kappa \lp \overline{U} \rp= - \ln \lp 1 - e^{-\kappa \lp U \rp }\rp
\]
is the {\em complement sparsity of $U$}. The {\em $l$-extension of $U$} is defined by
\[
Ext \lp U, l \rp = \lb t \in {[n] \choose l}~:~ \exists s \in U, s \subseteq t \rb,
\]
where the integer $l \in [n]$ is its {\em length}. That is, it is the family of $l$-sets each of which contains an $m$-set in $U$. If $l<m$, $Ext(U,l)$ is empty. The letter $V$ is used to express a sub-family of the $l$-extension, $i.e.$, $V \subseteq Ext \lp U, l \rp$.

\begin{example}
If $n=7$, $m=3$, $l=5$ and $U=\lb \lb 1,2,3 \rb, \lb 1,4,6\rb \rb$, then
\(
&&
[n]=\lb 1,2,3,4,5,6,7 \rb,
\\ &&
Ext(U, l) = \big\{ 
\lb 1, 2, 3, 4, 5 \rb,
\lb 1, 2, 3, 4, 6 \rb,
\lb 1, 2, 3, 4, 7 \rb,
\lb 1, 2, 3, 5, 6 \rb,
\lb 1, 2, 3, 5, 7 \rb,
\\ && \spca
\lb 1, 2, 3, 6, 7 \rb,
\lb 1, 2, 4, 5, 6 \rb,
\lb 1, 2, 4, 6, 7 \rb,
\lb 1, 3, 4, 5, 6 \rb,
\lb 1, 3, 4, 6, 7 \rb,
\lb 1, 4, 5, 6, 7 \rb
\big\}
\\ &&
\kappa(U) = \ln {n \choose m} - \ln |U| = \ln {7 \choose 3} - \ln 2=\ln 35/2,
\eqand
\\ &&
\kappa(Ext(U, l)) = \ln {n \choose l} - \ln |Ext(U, l)| = \ln {7 \choose 5} - \ln 11=\ln 21/11,
\)
$Ext(U, l)$ is the family of 5-sets each of which contains $\lb 1, 2, 3\rb \in U$ or $\lb 1,4, 6\rb \in U$.
\end{example}

\medskip

The $l$-extension together with {\em the $l$-shadow $\lb t \in {[n] \choose l}~:~ \exists s \in U, t \subseteq s \rb$} \cite{frankl}
and {\em $l$-neighborhood $\lb s \subseteq [n] ~:~ \exists t \in U, 
\left| s \setminus t \right| + |t\setminus s| \le l \rb$} of $U$ \cite{frankl2}
forms a significant mathematical structure in the Hamming space\footnote{
The $l$-shadow contains all the $l$-sets that are subsets of some $m$-sets in $U$. So it is empty if $l>m$. The $l$-neighborhood is the family of sets, each of which has Hamming distance $l$ or less from an element in $U$. 
}. Some facts are found previously \cite{my3} \cite{my5}. Especially, the following claim is shown with the $l$-extension \cite{my6}: An $n$-vertex graph $G$ contains at most $\lf {n \choose l} \cdot  2 \exp \lp - \frac{(l-1)k}{2n(n-1)}   \rp \rf$ cliques of size $l$, if the number of edges is $\frac{n(n-1)}{2}-k$ in $G$. Thus if $k$ edges are removed from the complete graph of $n$ vertices and $l \in [n]$ is an integer such that $lk \gg n^2$, the number of $l$-cliques in the remaining graph is much smaller than ${n \choose l}$.

A set $s \subseteq [n]$ is said to {\em generate} $t \subseteq [n]$ if $s \subseteq t$. Every $l$-set in $Ext \lp U, l \rp$ is generated by an $m$-set in $U$. We also say that $U$ {\em generates} $V$ if $V \subseteq Ext \lp U, l \rp$.

\subsection{Useful Formulas} \label{UsefulFormulas}

Assume that the size $n$ of the universal space grows to infinity. Any objects such as numbers, sets and families are actually functions of $n$. For example, if we say $m \in [n]$, it is a function $m: \Z_{>0} \rightarrow \Z_{>0}$ of $n$ such that $m(n) \in [n]$. A {\em constant} is a positive real number whose value is the same for all $n$. The letter $\gamma$ denotes a constant, and $\ep$ a sufficiently small constant that may depend on $\gamma$.

For two non-negative real numbers $f$ and $g$, write $f = O \lp g \rp$ if there exists a constant $\gamma$ such that $f \le \gamma g$ for every sufficiently large $n$, and $f = \Omega \lp g \rp$ if $g=O\lp f \rp$. Write $f = o \lp g \rp$ or $f \ll g$ if $\lim_{n \rightarrow \infty} f / g = 0$. In addition, $f =\Theta \lp g \rp$ means $f=O \lp g \rp$ and $f= \Omega \lp g \rp$. Such an order notation may express a function of the specified growth rate. For example, if we say that $m= \sqrt n - o \lp \sqrt n \rp$, it means $m$ equals $\sqrt n - q$ for a non-negative real number $q \ll \sqrt n$. An obvious floor or ceiling function is omitted.

A family $U$ of $m$-sets is said to be a {\em minority} if $\kappa \lp U \rp \gg 1$ in the considered space, and {\em majority} if $\kappa \lp \overline{U} \rp \gg 1$. We have the Taylor series
\[
\ln \lp 1 + x \rp= \sum_{j \ge 1} \frac{- (-x)^j}{j},
~~~
 x \in \lp -1, 1 \rbr,
\]
of natural logarithm. If $\kappa \lp U \rp \ge \lambda$ for a real number $\lambda \gg 1$, 
\[
\kappa \lp \overline{U} \rp \le - \ln \lp 1 - e^{-\lambda }\rp 
= e^{-\lambda} + O \lp e^{-2 \lambda} \rp
= e^{-\lambda + o(1)} \ll 1, 
\]
by the Taylor series, $i.e.$, if $U$ is a minority, its complement is a majority in $[n]$.

The following two well-known identities on Binomial coefficients are especially useful in this paper.

\beeq{basic1}
{p \choose r}{p - r \choose q-r} = {p \choose q}{q \choose r}
\eeq

\beeq{basic1-2}
\sum_{j} {p - r \choose j}{r \choose q-j} = {p \choose q}
\eeq
For example, suppose $U \subseteq {[n] \choose m}$ generates $V \subseteq Ext \lp U, l \rp$. Since each $m$-set in $U$ generates at most ${n-m \choose l-m}$ $l$-sets in $V$, we have with \refeq{basic1}
\beeqn
\nonumber
&&
|U| \ge \frac{|V|}{{n-m \choose l-m}}
=
\frac{{n \choose l}e^{-\kappa \lp V \rp}}{{n-m \choose l-m}}
=\frac{{n \choose l}{n \choose m}e^{-\kappa \lp V \rp}}{{n \choose m}{n-m \choose l-m}}
=\frac{{n \choose m}e^{-\kappa \lp V \rp}}{{l \choose m}},
\sothat \label{SparsityBoundForUGeneratingV}
\kappa \lp U \rp \le  \ln {l \choose m} + \kappa \lp V \rp.
\eeqn

We express summations the same way as \cite{knuth}. In \refeq{basic1-2}, for instance, regard ${p-r \choose j}=0$ if $j \not\in [p-r] \cup \lb 0 \rb$, where the unconstrained index $j$ denotes every integer $j \in \Z$.

\subsection{Asymptotics on Binomial Coefficients} \label{BinomialAsymptotics}

Define the function $S : (0, 1) \rightarrow \R$ by
\beeq{DefOfSx}
S \lp x \rp = \sum_{j \ge 1} \frac{x^j}{j(j+1)}.
\eeq
The following lemma approximates the natural logarithm of the binomial coefficient ${p \choose q}$. Its proof is found in Appendix.

\begin{lemma} \label{asymptotic}
\(
&&
\left| \ln {p \choose q} - q \lp \ln \frac{p}{q} + 1 - S \lp \frac{q}{p} \rp  \rp
- \frac{1}{2} \ln \frac{p}{2 \pi q(p-q)}
\right| = O \lp \frac{1}{\min \lp q, p-q \rp} \rp,
\)
for $p, q \in \Z$ such that $0<q < p$. \qed
\end{lemma}

\medskip

It means 
$
\left| \ln {p \choose q} - q \lp \ln \frac{p}{q} + 1 - S \lp \frac{q}{p} \rp  \rp
- \frac{1}{2} \ln \frac{p}{q(p-q)}
\right| = O \lp 1 \rp.
$
By our definition of order notations, we may write it as
\[
\ln {p \choose q} = q \lp \ln \frac{p}{q} + 1 - S \lp \frac{q}{p} \rp  \rp
+ \frac{1}{2} \ln \frac{p}{q(p-q)} \pm O \lp 1 \rp.
\]
Here $\pm O \lp 1 \rp$ is $r$ or $-r$ for a non-negative real number $r$ such that $r = O \lp 1 \rp$. It expresses a possibly negative error whose absolute value is bounded by a constant.

Also, if $p$ and $q$ are positive real numbers, omitted floor/ceiling functions may add an extra $\pm O \lp \ln (p+q) \rp$ error. For example, it is straightforward to see\\ $\ln {n \choose \lf \sqrt n \rf} = \sqrt n \lp \ln \frac{n}{\sqrt n} + 1 - S \lp \frac{\sqrt n}{n} \rp \rp \pm O \lp \ln n \rp$ from the theorem since $0 \le \sqrt n - \lf \sqrt n \rf<1$.

Noting the above, we have:

\begin{corollary}
\beeq{Proportional}
\ln {p+q \choose r}
=
\ln {p \choose \frac{rp}{p+q} }{q \choose \frac{rq}{p+q} } 
\pm O \lp \ln \lp p+ q \rp \rp.
\eeq
\end{corollary}
\begin{proof}
\(
\ln {p \choose \frac{rp}{p+q} }{q \choose \frac{rq}{p+q} } 
&=&
\lp \frac{rp}{p+q}  + \frac{rq}{p+q} \rp \lp \ln \frac{p+q}{r} + 1 - S \lp \frac{r}{p+q}\rp \rp
\pm O \lp \ln \lp p+q \rp \rp
\nexteqline
\ln {p+q \choose r} \pm O \lp \ln \lp p+q \rp \rp.
\)
\end{proof}

\medskip \medskip

We will also find the following lemmas useful.

\medskip

\begin{lemma} \label{basic2}
\[
{n-m \choose l}={n \choose l} e^{-\frac{lm}{n} - o\lp \frac{lm}{n} \rp},
\]
for $m, l \in [n]$ such that $l+m \ll n$. \qed
\end{lemma}

\medskip

\begin{lemma} \label{basic3}
\[
\ln {l-m \choose m-j}{m \choose j}
\le
\ln {l \choose m}  - j \ln \frac{j l}{m^2} + j  + \ln j+ O(1),
\]
for $l, m \in [n]$ and $j \in [m]$ such that $m^2 \le l$. 
\qed
\end{lemma}

\medskip

\noindent
Their proofs are presented in Appendix also.

\subsection{Some Properties of the $l$-Extension}

\subsubsection{Marks and Double Marks} \label{Marks}

Let $U \subseteq {[n] \choose m}$ and its $l$-extension $V=Ext \lp U, l \rp$ be fixed so that $l>m$. A {\em mark} is a pair $(t, d)$ such that $t \in U$, $d \in V$ and $t \subset d$, and a {\em double mark} a triple $(t, t', d)$ such that $t, t' \in U$, $d \in V$ and $t \cup t' \subset d$. The families of marks and double marks are denoted by ${\cal M}$ and ${\cal D}$, respectively. Their sparsities are defined as
\[
\kappa\lp \cal M \rp = - \ln \frac{|{\cal M}|}{{n \choose l}{l \choose m}}
\eqand
\kappa \lp \cal D \rp = - \ln \frac{|{\cal D}|}{{n \choose l}{l \choose m}^2},
\]
resp., $i.e.$, $|{\cal M}| ={n \choose l}{l \choose m} e^{-\kappa \lp \cal M \rp}$ and $|{\cal D}|={n \choose l}{l \choose m}^2 e^{-\kappa \lp D \rp}$.

\begin{note}
Both ${\cal M}$ and ${\cal D}$ are families of tuples of fixed cardinality subsets of $[n]$. Assume that the maximum size $N$ of such a family ${\cal F}$ is clearly defined; for example, the maximum number of double marks $(t, t', d)$ is $N={n \choose l}{l \choose m}^2$. The {\em sparsity} of ${\cal F}$ is defined by $\kappa \lp {\cal F} \rp= \ln N - \ln |{\cal F}|$, which is consistent with the aforementioned definition $\kappa\lp U \rp= \ln {n \choose m} - \ln |U|$.  This leads to the above $\kappa\lp {\cal M} \rp$ and $\kappa \lp {\cal D}\rp$. The {\em complement sparsity} of ${\cal F}$ is $\kappa\lp \overline{\cal F} \rp = \ln N - \ln \lp N - |{\cal F}|\rp$.  The family ${\cal F}$ is a {\em majority} if $\kappa \lp \overline{\cal F}\rp \gg 1$.
\end{note}

$\kappa \lp {\cal M} \rp$ is equal to $\kappa \lp U \rp$. Since each $m$-set $t \in U$ produces ${n-m \choose l-m}$ marks, there are exactly
\[
|{\cal M}|=|U| {n-m \choose l-m} = {n \choose m} e^{-\kappa\lp U \rp} {n-m \choose l-m}
={n \choose l}{l \choose m} e^{-\kappa\lp U \rp},
\]
marks by \refeq{basic1}. It equals ${n \choose l}{l \choose m} e^{-\kappa \lp \cal M \rp}$ so $\kappa \lp \cal M \rp= \kappa \lp U \rp$.


\begin{figure}
\centering
\includegraphics[width=100mm]{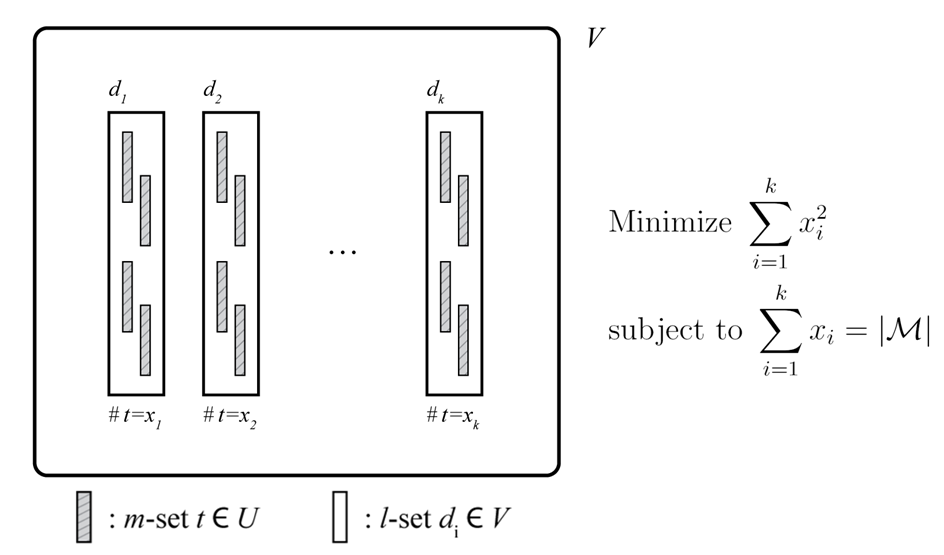}  
\caption{Relation between the Numbers of Marks and Double Marks} \label{fig1}
\end{figure}

An $l$-set $d \in V$ is incident\footnote{Given a family of tuples $s=\lp s_1, s_2, \ldots, s_q\rp$, let $s'$ be any $s_i$ or tuple of some $s_i$ (a {\em projection of $s$}), and $s''$ be the tuple of the remaining components. For these $s$, $s'$ and $s''$, we say that $s'$ is {\em incident to} $s$ and $s''$.} to at most ${l \choose m}$ marks. There must be  ${n \choose l} e^{-\kappa \lp \cal M \rp}={n \choose l} e^{-\kappa(U)}$ or more $l$-sets in $V$. Thus $\kappa \lp V \rp \le \kappa \lp U \rp$, $i.e.$, the $l$-extension $V$ is at least as dense as $U$. We can improve the bound with $\kappa \lp \cal D \rp$. Let us observe the following lemma. 

\begin{lemma} \label{1}
$\kappa\lp U \rp \le \kappa \lp \cal D \rp \le 2 \kappa \lp U \rp- \kappa \lp Ext(U,l) \rp$.
\end{lemma}
\begin{proof}
Let $V =Ext \lp U, l \rp = \lb d_1, d_2, \ldots, d_k \rb$, and $x_i$ for $i \in [k]$ be the number of marks incident to $d_i$. Regarding $x_i$ as variables, solve the following optimization problem with a Lagrange multiplier: minimize $\sum_{i=1}^k x_i^2$ subject to $\sum_{i=1}^k x_i = |{\cal M}|$ (\reffig{fig1}). Conclude $\sum_{i=1}^k x_i^2 \ge \lp \sum_{i=1}^k x_i \rp^2 \Big/ k$. Thus,
\(
&&
|V| \cdot |{\cal D}| \ge |{\cal M}|^2,
\sothat
|V| {n \choose l} {l \choose m}^2 e^{-\kappa \lp D \rp}
\ge
\lp {n \choose l} {l \choose m} e^{-\kappa \lp U \rp} \rp^2,
\sothat
|V| \ge {n \choose l} e^{- 2\kappa \lp U \rp + \kappa_{\cal D}},
\sothat
\kappa \lp V \rp \le 2 \kappa \lp U \rp - \kappa \lp \cal D \rp.
\)
The upper bound $\kappa \lp \cal D \rp \le 2 \kappa \lp U \rp- \kappa \lp V \rp=2 \kappa \lp U \rp- \kappa \lp Ext(U,l) \rp$ follows.

To show the lower bound $\kappa \lp \cal D \rp \ge \kappa\lp U \rp$, we maximize $\sum_{i=1}^k x_i^2$ subject to $\sum_{i=1}^k x_i = |{\cal M}|$ to conclude $|{\cal D}| \le {n \choose l}{l \choose m}^2 e^{-\kappa \lp U \rp}$. 
 \end{proof}

Put 
\[
\kappa_{\cal D} \stackrel{def}{=} \kappa \lp {\cal D} \rp - \kappa(U), 
\]
to call it the {\em proper sparsity of ${\cal D}$}. It satisfies
\beeq{SizeOfD}
|{\cal D}|={n \choose l}{l \choose m}^2 e^{-\kappa \lp U \rp - \kappa_{\cal D}}.
\eeq
By the above lemma,
\beeq{kappaD}
0 \le \kappa_{\cal D} \le  \kappa \lp U \rp- \kappa \lp Ext(U, l) \rp.
\eeq
The implied inequality $\kappa \lp V \rp =\kappa \lp Ext(U,l) \rp \le \kappa \lp U \rp- \kappa_{\cal D}$ is the improvement over the aforementioned simple bound $\kappa \lp V \rp\le \kappa \lp U \rp$.

\subsubsection{Sub-Family of $U$ in a Sphere} \label{Shpere}

The {\em sphere $S_{t,j}$ of radius $j \in[m]$ about an $m$-set $t$} is the family of $t' \in {[n] \choose m}$ such that $|t' \setminus t| = j$. The {\em sub-family of $U$ in the sphere $S_{t,j}$} is 
\[
{\cal S}(t, j)=U \cap S_{t,j}. 
\]
Its sparsity $\kappa \lp {\cal S}(t, j) \rp$ satisfies $\left| {\cal S}\lp t, j \rp \right|={n-m \choose j}{m \choose m-j}e^{-\kappa \lp {\cal S}(t, j) \rp}$.

The relationship between $\kappa \lp {\cal S}(t, j) \rp$ and $\kappa \lp {\cal D}\rp$ has an interesting property. For each $t \in U$ and $t' \in {\cal S}\lp t, m- j \rp$, its union $t \cup t'$ has size $2m-j$. They create exactly ${n - 2m+j  \choose l - 2m+j}$ double marks. With \refeq{basic1}, we have
\(
\left| \cal D \right| &=&
\sum_{\scriptstyle t \in U \atop 0 \le j \le m} 
\left| {\cal S}\lp t, m- j \rp \right| {n - 2m+j  \choose l - 2m+j}
\nexteqline
\sum_{t \in U, j} 
{n-m \choose m-j}{m \choose j}e^{-\kappa \lp {\cal S}(t, m-j) \rp}
{n - m - (m-j)  \choose l - m - (m-j)}
\nexteqline
{n-m \choose l-m}
\sum_{t \in U, j} 
{l-m \choose m-j}{m \choose j}e^{-\kappa \lp {\cal S}(t, m-j) \rp}.
\)

Notice that $\sum_{j} {l-m \choose m-j}{m \choose j}= {l \choose m}$ by \refeq{basic1-2}. Let $\kappa_{\cal S}$ be the average of $\kappa \lp {\cal S}(t, j) \rp$ over all $t \in U$ and $j \in [m]$ with respect to this summation, $i.e.$, $\sum_{t \in U, j} 
{l-m \choose m-j}{m \choose j}e^{-\kappa_{\cal S}(t, m-j)}
={n \choose m}e^{-\kappa \lp U \rp}{l \choose m}e^{-\kappa_{\cal S}}.
$
Then 
\(
\left| \cal D \right| &=&
{n-m \choose l-m}
{n \choose m}e^{-\kappa \lp U \rp}{l \choose m}e^{-\kappa_{\cal S}}
=
{n \choose l}{l \choose m}^2e^{-\kappa \lp U \rp - \kappa_{\cal S}}.
\)
Therefore, $\kappa_{\cal D}=\kappa_{\cal S}$ by \refeq{SizeOfD}, $i.e.$, the proper sparsity of the double mark family ${\cal D}$ equals the average sparsity of the sub-families of $U$ in the spheres.

With \refeq{kappaD}, we have $0 \le \kappa_{\cal S} \le \kappa(U) - \kappa \lp Ext(U, l)\rp$. The average sparsity $\kappa_{\cal S}$ is upper-bounded by $\kappa(U)$. As it approaches $\kappa(U)$, the sparsity of the $l$-extension gets closer to zero. In other words, we have a denser extension with sparser sub-families in the spheres. This is the key observation to prove the main claim of  \refsec{ExtensionGenerator}: the extension generator theorem.

\subsubsection{Space-Augmenting Extension} \label{Augmenting}

Later, we will encounter a situation where we want to expand the considered space containing a target $l$-extension. The following technique will be useful: Let $U \subseteq {[n] \choose m}$ and $X$ be a set not intersecting $[n]$. The {\em space-augmenting extension $V$ of $U$ from the space $[n]$ into $[n] \cup X$} is the $\lp m + |X| \rp$-extension of  $U$ in $[n] \cup X$, $i.e.$,
$
V = \lb b \in {[n] \cup X \choose m + |X|}~:~ \exists s \in U, s \subset b \rb
$. Show the following lemma. 

\begin{lemma} (Space-Augmenting Extension) \label{SpaceAugmentingExtension}
Let $U \subseteq {[n] \choose m}$, and $X$ be a set not intersecting $[n]$. The sparsity of the space-augmenting extension $V$ of $U$ from the space $[n]$ into $[n] \cup X$ is at least $\kappa \lp U \rp$.
\end{lemma}
\begin{proof}
Any element in $V$ is constructed by concatenating $t \in Ext \lp U, m+j \rp$ in the space $[n]$ with an $(|X|-j)$-set in $X$, for some $j \in \lbr |X| \rbr \cup \lb 0 \rb$. Thus
$
|U| \ge \sum_{j} {n \choose m+j} e^{-\kappa \lp U \rp} {|X| \choose |X|- j}
={n+|X| \choose m+|X|} e^{-\kappa \lp U \rp}
$
by \refeq{basic1-2}. The lemma follows.
\end{proof}

\subsection{Split of a Family} \label{SplitOfFamily}

A {\em split
\footnote{
A split is a set of ordered tuples whose components are fixed cardinality subsets of $[n]$. A {\em partition}, well-known in combinatorics and computer science, is a similar notion, which is an unordered family of sets with no cardinality constraint. Despite their similarity, splits give different flavors from partitions in their counting and operations with \refeq{CardinalityN}. We call such tuples splits in this paper. 
}
of $U$} is a family $W$ of $q$-tuples $\lp s_1, s_2, \ldots, s_q \rp$ such that:
\begin{enumerate}[(i)]
\item $\bigcup_{j=1}^q s_j \in U$ for every $\lp s_1, s_2, \ldots, s_q \rp \in W$, 
\item $s_1, s_2, \ldots, s_q$ are pairwise disjoint, and
\item for each $j \in [q]$, the cardinality of $s_j$ is a fixed positive integer. 
\end{enumerate}

\noindent
We also say that $\lp s_1, s_2, \ldots, s_q \rp$ is a {\em split of} the set $\bigcup_{j=1}^q s_j$ if $s_j$ are non-empty and pairwise disjoint.

The maximum size of such $W$ is
\beeq{CardinalityN}
N = {n \choose |s_1|} {n-|s_1| \choose |s_2|} 
{n-|s_1| -|s_2| \choose |s_3|} 
\cdots 
{n-|s_1|-|s_2|-\cdots - |s_{q-1}| \choose |s_q|}. 
\eeq
This defines the sparsity of $W$ as $\kappa \lp W \rp= \ln N - \ln |W|$.

It can be shown that $\kappa \lp W \rp= \kappa \lp U \rp$ if $W$ is the family of all $(s_1, s_2, \ldots, s_q)$ such that (i)--(iii). For example, if $q=2$, there are ${m \choose |s_1|}={|s_1|+|s_2| \choose |s_1|}$ splits $\lp s_1, s_2 \rp$ such that $s_1 \cup s_2$ is a given fixed element in $U \subseteq {[n] \choose m}$. The total number of $(s_1, s_2) \in W$ is
$
{|s_1|+|s_2| \choose |s_1|} \left| U \right| = {|s_1|+|s_2| \choose |s_1|} {n \choose |s_1|+|s_2|}e^{-\kappa \lp U \rp}
=N e^{-\kappa \lp U \rp}.
$
Thus $\kappa \lp W \rp= \kappa \lp U \rp$.

\section{The Extension Generator Theorem} \label{ExtensionGenerator}

\pdffigure{fig2}{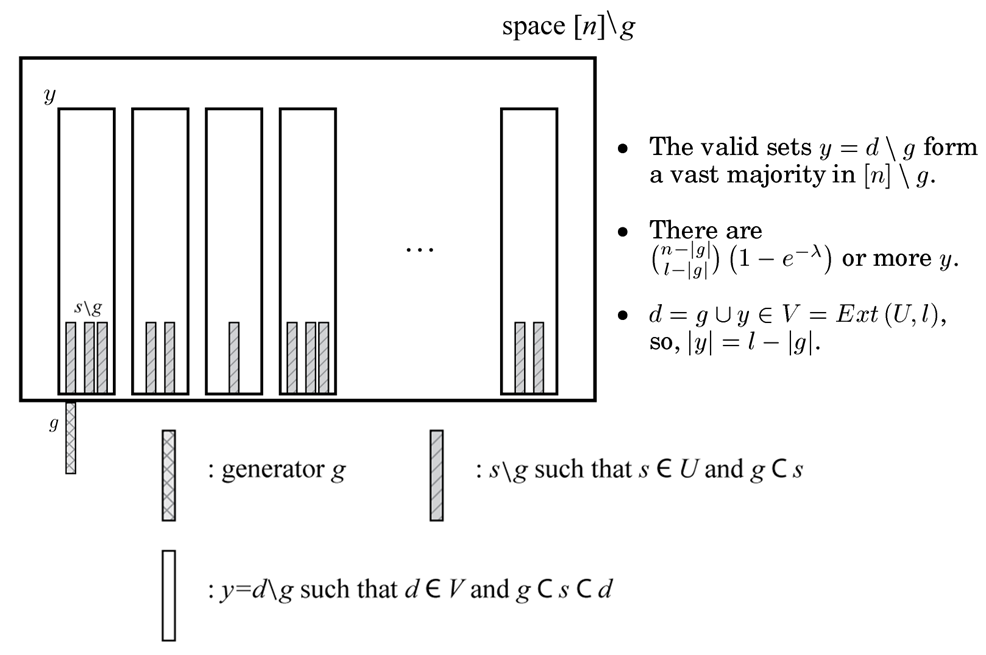}{An $\lp l, \lambda \rp$-Extension Generator $g$ Such That $\lambda \gg 1$}

Let
\[
U_g \stackrel{def}{=} \lb s \in U ~:~ s \supseteq g \rb,
\]
for $U \subseteq {[n] \choose m}$ and $g \subset [n]$, $i.e.$, $U_g$ is the family of $m$-sets in $U$ generated by the set $g$. An {\em $(l, \lambda)$-extension generator of $U$} is a set $g \subset [n]$ such that
\[
\left| Ext\lp U_g , l  \rp \right| \ge {n-|g| \choose l-|g|} \lp 1 - e^{-\lambda} \rp.
\]
Here $l \in [n] \setminus [m]$ is the {\em extension length} and $\lambda \in \R_{\ge 0}$ the {\em complement sparsity of $g$}. 

Denote by $y$ an $\lp l-|g|\rp$-set in the space $[n] \setminus g$. We may write it as $y \in {[n]\setminus g \choose l-|g|}$ according to the definition given in \refsec{family}. We say that $y$ is a {\em valid set of $g$} if $g \cup y\in Ext \lp U_g, l \rp$, and {\em error set of $g$} otherwise.

Define
\beeq{HatUg}
\hat{U}_g \stackrel{def}{=} \lb s \setminus g ~:~ s \in U_g \rb.
\eeq
It is a family of $\lp m-|g|\rp$-sets in the space $[n]\setminus g$. If $g$ is an $(l, \lambda)$-extension generator of $U$, 
\[
\left| Ext\lp \hat{U}_g , l - |g|  \rp \right| \ge {n-|g| \choose l-|g|} \lp 1 - e^{-\lambda} \rp.
\]
Its complement sparsity $\kappa \lp \overline{Ext \lp \hat{U_g} , l-|g| \rp} \rp$ is at least $\lambda$ in the space $[n]\setminus g$. When $\lambda \gg 1$, the extension generator $g$ is a set such that $Ext \lp \hat{U_g} , l-|g| \rp$ is a majority. It is illustrated in \reffig{fig2}.

In this section we develop a theorem that guarantees the existence of a small extension generator $g$ for certain $U$, $l$ and $\lambda \gg 1$. It is the aforementioned structural theorem that generalizes the sunflower lemma in most cases. Our process consists of two parts to find such $g$: {\em Phases I and II}. We show them in the following subsections.

\subsection{Phase I}

For given $U\subseteq {[n] \choose m}$ such that $m^2 \ll n$, let an integer $l_0 \in [n]$ satisfy
\beeq{PhaseIContstraint}
m^2 < l_0.
\eeq
In Phase I, we prove the existence of an $\lp l_0, \lambda_0 \rp$-extension generator $g$ of size at most $\kappa \lp U \rp {\Big /} \ln \frac{l_0}{m^2}$ for a real number $\lambda_0 = \Omega \lp 1 \rp$.

The faimly $\hat{U}_g$ in \refeq{HatUg} and its sparsity $\kappa \lp \hat{U}_g \rp$ are defined in the space $[n] \setminus g$. Our $g$ is a maximal subset of $[n]$ such that
\beeq{KappaHatUg}
\kappa \lp \hat{U}_{g}  \rp \le \kappa \lp U\rp - r |g|,
\eqwhere
r \stackrel{def}{=} \ln \frac{l_0}{m^2}.
\eeq
Such a set $g \subset [n]$ has a size bounded by
\beeq{SizeOfGenerator}
|g| \le \frac{\kappa \lp U \rp}{r},
\eeq
since $\kappa \lp \hat{U}_g \rp$ is non-negative. There exists a maximal set $g \subset [n]$ such that \refeq{KappaHatUg} and \refeq{SizeOfGenerator}.

In the rest of the subsection, we prove that g is an $\lp l_0, \lambda_0 \rp$-extension generator of $U$ for some $\lambda_0 = \Omega \lp 1 \rp$. The arguments will not be affected by the set $g$. Assume 
\beeq{Simple}
g=\emptyset
\eqand
\hat{U}_g=U,
\eeq
for simplicity. The proof will use
\beeqn
\nonumber
\ln {l_0-m \choose m-j}{m \choose j}
&\le&
\ln {l_0 \choose m}  - j \ln \frac{j l_0}{m^2} + j  + \ln j+ O(1)
\\ &\le&
\label{L24Modified}
\ln {l_0 \choose m}  - j \lp \ln j  + r  \rp + j + \ln j+ O(1),
\eeqn
for $j \in [m]$, seen with \reflm{basic3}. The inequality holds true\footnote{
If $g \ne \emptyset$, we have $\hat{U}_g$ in place of $U$. This changes $n$, $l_0$ and $m$ into $n-|g|$, $l_0 - |g|$, and $m -|g|$, respectively. The real number $l_0 / m^2$ is replaced by $(l_0-|g| \big/ (m-|g|)^2$ in \refeq{L24Modified}, while $r$ is independent of $g$ defined by \refeq{KappaHatUg}. It is straightforward to show $(l_0-|g|) \big/ (m-|g|)^2 \ge l_0 / m^2=r$ for any non-negative integer $|g|< m$. The inequality \refeq{L24Modified} is still true with the change. 
} whether it is in Case \refeq{Simple} or not.

\medskip\medskip

Observe the following lemma.

\begin{lemma}
Assume \refeq{Simple}. For each $s \in U$ and $j \in [m]$, let ${\cal S}(s, m-j)$ be the sub-family of $U$ in the sphere of radius $m-j$ about $s$ as defined in \refsec{Shpere}. The sparsity of ${\cal S} \lp s, m - j \rp$ is more than $\kappa \lp U \rp - j r- o(1)$.
\end{lemma}
\begin{proof}
We are given an $m$-set $s \in U \subseteq {[n] \choose m}$ and $j \in [m]$. Let $s'$ be a subset of $s$ of size $j$, written as $s' \in {s \choose j}$. Observe that
\[
|U_{s'}| < {n -j \choose m - j} \exp \lp -\kappa \lp U \rp + j r \rp,
\]
$i.e.$, the number of $t\in U$ such that $t \supseteq s'$ is upper-bounded as above. Otherwise, $\kappa \lp \hat{U}_{s'} \rp \le \kappa \lp U \rp - j r$ so that $g = s'$ satisfies \refeq{KappaHatUg}. The existence of $s' \in {s \choose j}$ such that $|U_{s'}| \ge {n -j \choose m - j} \exp \lp -\kappa \lp U \rp + j r \rp$ would contradict the maximality of $g$.

The observation leads to
\beeq{eqL31}
\left| {\cal S}\lp s, m - j \rp \right|
<{n-m \choose m-j}{m \choose j} \exp \lp -\kappa \lp U \rp + j r + o(1) \rp. 
\eeq
For each $s' \in {s \choose j}$, the number of $t' \in {\cal S}\lp s, m - j \rp$ such that $t' \supseteq s'$ is no more than $|U_{s'}|$. With \reflm{basic2} and the given condition\footnote{
The condition $m^2 \ll n$ is used to see $O \lp \frac{(m-j)^2}{n-j} \rp=o(1)$. If it is not in Case \refeq{Simple}, the value becomes $O \lp \frac{(m-j-|g|)^2}{n-j-|g|} \rp$. It is $o(1)$ for any $j$ and $|g|$. 
}
$m^2 \ll n$, we have ${n-m \choose m-j}={(n-j)-(m-j) \choose m-j}={n-j \choose m-j}\exp \lp - O \lp \frac{\lp m-j \rp^2}{n-j} \rp \rp={n-j \choose m-j}\exp \lp - O \lp m^2 / n \rp \rp = {n-j \choose m-j}\exp \lp - o(1) \rp$. There are less than 
\[
{n -j \choose m - j} \exp \lp -\kappa \lp U \rp + j r \rp
=
{n -m \choose m - j} \exp \lp -\kappa \lp U \rp + j r + o(1)  \rp
\]
elements $t' \in {\cal S}\lp s, m - j \rp$ such that $t' \supseteq s'$. Since there are ${m \choose j}$ $j$-sets $s' \in {s \choose j}$, the size of ${\cal S}\lp s, m - j \rp$ is bounded by \refeq{eqL31}.

Therefore, the sparsity of ${\cal S}\lp s, m - j \rp$ exceeds $\kappa \lp U \rp - j r- o(1)$.
 \end{proof}

\medskip

Each $s \in U$ and $t \in {\cal S}\lp s, m - j \rp$ create exactly ${n - (2m-j) \choose l_0 - (2m-j)}$ double marks since $|s \cup t|=2m-j$. Due to the lemma, 
$
\left| {\cal S} \lp s, m - j \rp \right| < {n-m \choose m-j}{m \choose j}
e^{- \kappa \lp U \rp + j r +o(1)}
$.
We have the following upper bound on the size of ${\cal D}$:
\beeqn \label{DoubleMarkTransform}
\left| {\cal D} \right| &=&
\sum_{\scriptstyle s \in U,~0 \le j \le m \atop t \in {\cal S} \lp s, m - j \rp}
{n - (2m-j) \choose l_0 - (2m-j)}
\\ &<& \nonumber
\sum_{s \in U, j}
{n-m \choose m-j}{m \choose j}
e^{- \kappa \lp U \rp + j r +o(1)}
{n - 2m+j \choose l_0 - 2m+j}
~~~~~\lp \textrm{\reflm{1}}\rp
\nexteqline \nonumber
e^{- \kappa \lp U \rp  + o(1)}
\sum_{s \in U, j}
{n-m \choose l_0-m } {l_0-m \choose m-j}{m \choose j}
e^{j r}
~~~~~~~~~~~\lp \textrm{Remark (i) below}\rp
\nexteqline \nonumber
{n \choose m}{n-m \choose l_0-m }
e^{- 2 \kappa \lp U \rp  + o(1)}
\sum_{j}
{l_0-m \choose m-j}{m \choose j}
e^{j r}
~~~~~\lp |U|= {n \choose m}e^{-\kappa \lp U \rp}\rp
\nexteqline \nonumber
{n \choose l_0}{l_0 \choose m}
e^{- 2 \kappa \lp U \rp  + o(1)}
\sum_{j}
{l_0-m \choose m-j}{m \choose j}
e^{j r}
~~~~~~~~~\lp \textrm{By \refeq{basic1}}\rp
\nexteqline \nonumber
{n \choose l_0}{l_0 \choose m}^2
e^{- 2 \kappa \lp U \rp +  O \lp 1 \rp }.
\spcd  \spcd \spcd \lp \textrm{Remark (ii) below}\rp
\eeqn
Here the two remarks are:\\
i)  ${n-m \choose m-j}{n -2 m+j \choose l_0 - 2m + j}
={n-m \choose m-j}{ (n -m) - (m-j)  \choose (l_0 - m)  - (m - j)}
={n-m \choose l_0-m}{l_0-m \choose m-j}$ by \refeq{basic1}, and\\
ii) 
\(
&&
\sum_{j} {l_0-m \choose m-j}{m \choose j} e^{j r}
<
{l_0 \choose m} + \sum_{j \ge 1} {l_0-m \choose m-j}{m \choose j} e^{j r}
\\ &\le&
{l_0 \choose m} +\sum_{j \ge 1} {l_0 \choose m} e^{- j \lp \ln j + r \rp + j + \ln j + O(1) + j r}
~~(\textrm{by \refeq{L24Modified}})
\nexteqline
{l_0 \choose m} +\sum_{j \ge 1} {l_0 \choose m} e^{- (j-1) \ln j + j + O(1)}
\\ &\le&
{l_0 \choose m} e^{O(1)} 
\lp 1+ e^{-0 \ln 1 + 1} + e^{-1 \ln 2 + 2} + e^{-2 \ln 3+3} + e^{-3 \ln 4 +4} + \cdots \rp
\nexteqline
{l_0 \choose m} e^{O \lp 1 \rp}.
\)

\medskip

\noindent
We have seen  $\left| {\cal D} \right| < {n \choose l_0}{l_0 \choose m}^2
e^{- 2 \kappa \lp U \rp +  O \lp 1 \rp }$, that is, $\kappa \lp {\cal D} \rp> 2 \kappa \lp U \rp- O(1)$. By \reflm{1}, 
\(
&&
2 \kappa \lp U \rp- O(1) < \kappa \lp {\cal D} \rp
\le  2 \kappa \lp U \rp - \kappa \lp  Ext(U, l_0) \rp,
\sothat 
\kappa \lp E(U, l_0) \rp = O(1)
~~\Rightarrow~~\kappa \lp \overline{Ext \lp U, l_0 \rp} \rp=\Omega \lp 1 \rp. 
\)
In other words, the complement sparsity of the obtained generator $g$ is $\Omega \lp 1 \rp=\lambda_0$, satisfying the desired property of Phase I.

\medskip

In summary, we have proven the following lemma.

\begin{lemma} \label{Phase I}
Let $U\subseteq {[n] \choose m}$ and $l_0 \in [n]$ such that $m^2 \ll n$ and $m^2<l_0$. There exists an $(l_0, \Omega \lp 1 \rp)$-extension generator of $U$ whose size is at most $\kappa \lp U \rp {\Big /} \ln \frac{l_0}{m^2}$. \qed
\end{lemma}

\medskip \medskip

If $l_0 \gg m^2$ and $\kappa \lp U \rp$ is not large, the constructed generator $g$ is small satisfying \refeq{SizeOfGenerator}. We also have:
\(
&&
\left| Ext\lp U_g , l_0  \rp \right| \ge {n-|g| \choose l_0-|g|} \lp 1 - e^{-\lambda_0} \rp
~~~\textrm{in the space $[n]$, and}, 
\\ &&
\left| Ext\lp \hat{U}_g , l_0 -|g| \rp \right| \ge 
{n-|g| \choose l_0-|g|} \lp 1 - e^{-\lambda_0} \rp
~~~\textrm{in the space $[n]\setminus g$.} 
\)
The sparsity of ${\cal D}$ has played the central role to see it.

\subsection{Phase II}

We now show that the obtained $\lp l_0, \lambda_0 \rp$-generator $g$ is indeed an $\lp l, \lambda \rp$-generator for $l \in [n]$ and $\lambda \in \R_{>0}$ such that 
\[
1 \ll \lambda \ll \frac{n}{m^2}
\eqand
m^2 \lambda \le l \le n.
\]
In this Phase II of the generator construction, we will improve the complement sparsity from $\lambda_0= \Omega \lp 1 \rp$ into $\lambda \gg 1$ by increasing the extension length from $l_0$ to $l$.

We will prove that $Ext \lp \hat{U}_g, i  \lp l_0  - |g| \rp \rp$ for each $i>0$ has complement sparsity at least $i\lambda_0$ in the space $[n] \setminus g$. Let us assume \refeq{PhaseIContstraint} such that $l_0=O \lp m^2 \rp$. The maximum value of the index $i$ is $\frac{l-|g|}{l_0-|g|} = \Omega \lp \frac{l}{m^2} \rp$.  The claim means that the complement sparsity is at least
\[
\frac{l-|g|}{l_0-|g|} \lambda_0= \Omega \lp \frac{l}{m^2} \rp \cdot \lambda_0
=
\Omega \lp \frac{m^2 \lambda}{m^2} \rp  \cdot \Omega \lp 1 \rp = \Omega \lp \lambda \rp.
\]
Then our goal of finding an $\lp l, \lambda \rp$-generator is almost achieved. The $\lp l_0, \lambda_0 \rp$-generator $g$ we found in Phase I will be shown as a desired $(l, \lambda)$-generator. 

\medskip \medskip

Let
\[
U_i = Ext \lp \hat{U}_g, i \lp l_0  - |g| \rp \rp,
\]
in the space $[n] \setminus g$, and prove $\kappa \lp \overline{U_i} \rp \ge i \lambda_0$ by induction on $i$. The basis $i=1$ is true by Phase I. Assume true for $i$ and prove true for $i+1$. 

Once again, we assume $g=\emptyset$ for simplicity; regard that $[n]$ instead of $[n] \setminus g$ is the space to include all the considered sets. This means $U_i = Ext \lp U, i l_0 \rp$ and that
the given induction hypothesis is $\kappa \lp \overline{U_i} \rp > i \lambda_0$ in the universal space $[n]$.

To show the induction step $\kappa \lp \overline{U_{i+1}} \rp > (i +1) \lambda_0$, consider the pairs $(s, b)$ such that $b \in {[n] \choose il_0}$ and $s \in {[n] \setminus b \choose l_0}$, $i.e.$, $b$ is any $il_0$-set and $s$ is an $l_0$-set disjoint from $b$. The pairs are the splits of all the $(i+1)l_0$-sets. (Splits are defined in \refsec{SplitOfFamily}.) An $il_0$-set $b$ is an element in either $U_i$ or $\overline{U_i}$. Observe the following lemma.

\begin{lemma} \label{Tj}
Suppose $g =\emptyset$ and fix any $b \in \overline{U_i}$. The sparsity of the family of $s \in {[n] \setminus b \choose l_0}$ such that $s \cup b \in U_{i+1}$ is at most $\kappa \lp U_1 \rp$ in the space $[n]\setminus b$.
\end{lemma}
\begin{proof}
Let
\[
T_j=\lb t \setminus b ~:~ t \in U_1 \textrm{~and~} |t\setminus b|=j \rb. 
\]
We claim there exists $j \in [l_0] \cup \lb 0 \rb$ such that $T_j$ has sparsity at most $\kappa \lp U_1 \rp$ in the space $[n] \setminus b$. Suppose not. Then $\kappa \lp T_j \rp >\kappa \lp U_1 \rp$ for every $j \in [l_0] \cup \lb 0 \rb$. Any $l_0$-set $t \in U_1$ is a $j$-set in $T_j$ joined with an $\lp l_0 - j \rp$-set in the space $b$. By \refeq{basic1-2},
\[
\left| U_1  \right| 
\le \sum_j \left| T_j \right| {|b| \choose l_0-j}
<\sum_{j} {n-|b| \choose j}  e^{- \kappa \lp U_1 \rp} \cdot {|b| \choose l_0-j}
={n \choose l_0} e^{- \kappa \lp U_1 \rp} = \left| U_1 \right|.
\]
A contradiction that $|U_1|<|U_1|$. There exists $j \in [l_0] \cup \lb 0 \rb$ such that $\kappa \lp T_j \rp \le \kappa \lp U_1 \rp$.

Extend this $T_j$ from length $j$ to $l_0$ in the space $[n] \setminus b$, $i.e.$, consider $Ext\lp T_j, l_0 \rp$. Since $\kappa \lp T_j \rp \le \kappa \lp U_1 \rp$, its $l_0$-extension satisfies the same sparsity upper bound. Thus $\kappa \lp Ext\lp T_j, l_0 \rp \rp \le \kappa \lp U_1 \rp$. Since every element $s \in Ext\lp T_j, l_0 \rp$ meets the property $s \cup b \in U_{i+1}$, it proves the lemma.
 \end{proof}

\medskip

If $b \in U_i$, every $s \in {[n]\setminus b \choose l_0}$ creates a split $(s, b)$ of $s \cup b \in U_{i+1}$. If $b \in \overline{U_i}$, the sparsity of $s$ such that $s \cup b \in U_{i+1}$ in the space $[n] \setminus b$ is at most  $\kappa \lp U_1 \rp$ by the lemma. The total number of splits $(s, b)$ of $s \cup b \in U_{i+1}$ is at least $\beta {n \choose |b|}{n-|b| \choose |s|}$ where 
\(
&&
\beta \ge \lp 1 - e^{-\kappa \lp \overline{U_i} \rp}\rp +  e^{-\kappa\lp \overline{U_i} \rp} \lp 1 - e^{-\kappa\lp \overline{U_1} \rp} \rp
= 1 - e^{-\kappa \lp \overline{U_i} \rp - \kappa \lp \overline{U_1} \rp} >  1 - e^{- (i+1) \lambda_0},
\)
by induction hypothesis. It is equal to
\(
&&
\beta {n \choose |b|}{n-|b| \choose |s|}
=\beta {n \choose |b|}{n-|b| \choose \lp |b|+|s| \rp - |b|}
\nexteqline
\beta {n \choose |b|+|s|}{|b|+|s| \choose |s|}
=\beta {n \choose (i+1)l_0}{(i+1)l_0 \choose l_0},
\)
by \refeq{basic1}. An $(i+1)l_0$-set in $U_{i+1}$ produces at most $(i+1)l_0 \choose l_0$ such splits $(s, b)$. Therefore,
\[
\left| U_{i+1} \right| \ge \beta  {n \choose (i+1)l_0} \ge {n \choose (i+1)l_0} \lp 1 - e^{- (i+1) \lambda_0} \rp, 
\]
meaning $\kappa \lp \overline{U_{i+1}} \rp \ge (i +1) \lambda_0$. This proves the induction step.

\medskip \medskip

We now have the main claim of this section:

\begin{theorem} (Extension Generator Theorem) 
\label{ExtensionGeneratorTheorem}
Let\\
i) $U\subseteq {[n] \choose m}$ for $m \in [n]$, \\
ii) $\lambda \in \R$ such that $1 \ll \lambda \ll \frac{n}{m^2}$, \\
iii) $\ep \in (0,1)$ be a sufficiently small constant, and\\
iv) $l \in [n]$ such that $l > m^2 \lambda \big / \ep$.\\
There exists an $\lp l, \lambda\rp$-extension generator of $U$, whose size is at most $\kappa \lp U \rp {\Big /} \ln \frac{\ep l}{m^2 \lambda}$.
\end{theorem}
\begin{proof}
Condition (iv) means $\frac{\ep' l}{\lambda} > m^2 \big/ \ep'$ where $\ep'=\sqrt \ep$ is a sufficiently small constant. Set $l_0 = \frac{\ep' l}{\lambda} > m^2 \big/ \ep'$. This satisfies \refeq{PhaseIContstraint} in addition to $m^2 \ll n$ by (ii). The obtained set $g$ is an $\lp l_0, \Omega \lp 1 \rp \rp$-extension generator of $U$, as well as $\lp \ep' l, \Omega \lp \lambda \rp \rp$-generator by the above Phase II. It is also an $\lp l , \lambda \rp$-generator by performing Phase II steps for extra $1/\ep'$ times.
 \end{proof}

\subsection{Application to Circuit Complexity of the Clique Problem}
\label{AppToClique}

We will apply \refth{ExtensionGeneratorTheorem} to a monotone circuit $C$ to compute the clique function {\sc CLIQUE}$_{n,k}$, which is expressed as
\beeq{CliqueFunction}
\textrm{\sc CLIQUE}_{n,k}
~\Leftrightarrow~
\bigvee_{c \in {[n] \choose k}}
\bigwedge_{e \in {c \choose 2}} X_e.
\eeq
Here we consider a graph with the vertex set $[n]=\lb 1, 2, \ldots, n \rb$, every $k$-clique $c \in {[n] \choose k}$ also regarded as a $k$-set in the space $[n]$, and edges $e \in {c \choose 2}$, $i.e.$, 2-sets in the clique $c$. The Boolean variable $X_e$ is true iff the edge $e$ exists in the graph. The arrow denotes the logical equivalence.

Let $C$ be a circuit to compute {\sc Clique}$_{n, k}$ and $\alpha$ be any of its nodes. We say that a $k$-clique $c$ is {\em generated at $\alpha$} if the existence of all the edges in ${c \choose 2}$ implies the truth of $\alpha$. Suppose that there are ${n \choose k}e^{-\lambda}$ $k$-cliques generated at $\alpha$. This real number $\lambda$ is the {\em sparsity of $c$}, or of the set of $c$ generated at $\alpha$.

Our approach to show the exponential monotone complexity finds an extension generator $g$ of $c$ at each $\alpha$. Related variables and their properties are:
\beeqn
&& \label{PreliminaryNotation}
k=\sqrt[4]{n}: \textrm{the clique size,}
\\ && \nonumber
\ep \in (0,1):  \textrm{a sufficiently small constant,}
\\&& \nonumber
|C| \le exp\lp n^{\ep^2} \rp
\textrm{~where $|C|$ is the number of nodes in $C$, the {\em circuit size},}
\\ && \nonumber
q=n^{5\ep}: \textrm{an integer parameter assumed to divide $n$,}
\\ && \nonumber
g: \textrm{$\lp n/q, k \rp$-extension generator of $c$ generated at a node $\alpha$,}
\\ && \nonumber
y: \textrm{valid set(s) of $g$, and,}
\\ && \nonumber
\lambda_c=n^\ep : \textrm{sparsity upper bound of $c$ generated at $\alpha$.}
\eeqn
By \refth{ExtensionGeneratorTheorem}, there exists an $\lp n/q, k \rp$-extension generator $g$ of $c$ at every node $\alpha$ where  ${n \choose k}e^{-\lambda_c}$ or more $k$-cliques $c$ are generated.

The valid sets $y$ of such a generator $g$ are $\lp l-|g| \rp$-sets in the space $[n]\setminus g$ where $l=n/q$. They form a majority with complement sparsity at least $k=\sqrt[4]{n}$, which is seen by the theorem. Our construction on $C$ will require a common scoped space for all the nodes in $C$ rather than $[n]\setminus g$ that depends on $g$ and $\alpha$.

We extend the space $[n]\setminus g$ into $[n]$ with a space-augmenting extension defined in \refsec{Augmenting}: Apply \reflm{SpaceAugmentingExtension} to the family of $y \in {[n]\setminus g \choose l-|g|}$ in such a way that $[n] \leftarrow [n]\setminus g$ and $X \leftarrow g$. After this, $y$ have size exactly $l$ forming a majority in the space $[n]$. By the lemma, the family of $y$ satisfies the same complement sparsity lower bound $k$. The modified valid sets $y$ still satisfy the same property that:
\begin{quote}
At a considered node $\alpha$ of $C$, for every valid $y\in {[n] \choose l}$, there exists a $k$-clique $c$ generated at $\alpha$ such that $g \subset c \subset g \cup y$.
\end{quote}
After the application of space-augmenting extension, the generator $g$ may intersect $y$. Similar arguments lead to a more general statement.

\begin{corollary} \label{ExtensionGeneratorTheorem2}
Let $U$, $m$, $l$, $\lambda$ and $\ep$ be as given in \refth{ExtensionGeneratorTheorem}. There exists $g \subset [n]$ of size at most $\kappa \lp U \rp {\Big /} \ln \frac{\ep l}{m^2 \lambda}$, and a family ${\cal Y}_U \subset {[n] \choose l}$ with complement sparsity at least $\lambda$, such that  for every $y \in {\cal Y}_U$, there exists $s \in U$ satisfying $g \subseteq s \subseteq g \cup y$.
\end{corollary}

\medskip

The above modification is performed at one node $\alpha$. Let ${\cal Y}$ be the family of $l$-sets $y$ valid for all $g$ at any nodes in $C$. Its complement sparsity is $\Omega \lp k \rp$. For, there are no more than ${[n] \choose l}e^{-k}$ error sets $y$ of $g$ at one particular $\alpha$. The total number of error sets is bounded by
\[
{n \choose l}e^{-k} \cdot |C| =
{n \choose l}e^{-n^{1/4}+n^{\ep}} = {n \choose l}e^{-\Omega \lp k \rp},
\]
since $\ep \in (0,1)$ is sufficiently small. The remaining ${n \choose l} \lp 1- e^{-\Omega \lp k \rp} \rp$ $l$-sets $y$ are valid for all $g$. Thus
\beeq{ComplOfCalY}
\kappa \lp \overline{\cal Y} \rp = \Omega \lp k \rp, 
\eeq
$i.e.$, the complement sparsity of ${\cal Y}$ is $\Omega\lp k \rp$.

We now consider
\[
\textrm{splits~}{\bf y}=\lp y_1, y_2, \ldots, y_q \rp \textrm{~of $[n]$ such that $y_i \in {\cal Y}$.} 
\]
By the definition in \refsec{SplitOfFamily}, the $n/q$-sets $y_1, y_2, \ldots, y_q$ in such a split ${\bf y}$ are mutually disjoint so that their union is $[n]$. The fixed cardinality $|y_j|=n/q$ divides $n$ as assumed above\footnote{
If $q$ does not divide $n$, we disregard the $n - \lf n/q \rf q$ largest numbered vertices in $[n]$. With $q=n^{5\ep}$, this changes no asymptotic property discussed in the paper. 
}.

Let ${\cal F}$ be the family of such splits ${\bf y}$. Its maximum possible cardinality is $N$ given in \refeq{CardinalityN}. With $|y_j|=n/q$, it is expressed as
\[
N = \prod_{j=1}^q {n-(j-1)n/q \choose n/q}.
\]
The sparsity and complement sparsity of ${\cal F}$ are defined by this $N$. We can see with \refeq{ComplOfCalY} that ${\cal F}$ forms a majority with complement sparsity $\Omega \lp k \rp$: The number of $q$-tuples $(y_1, y_2, \ldots, y_q)$ such that $y_1 \not \in {\cal Y}$ and $y_1, y_2, \ldots, y_q$ are pairwise disjoint is exactly
\[
|{\cal Y}| \prod_{j=2}^q {n-(j-1)n/q \choose n/q} = {n \choose n/q }e^{-\kappa \lp {\cal Y} \rp} \prod_{j=2}^q {n-(j-1)n/q \choose n/q} = N e^{-\kappa \lp {\cal Y} \rp}. 
\]
By symmetry, it is true for the other $y_j$. The number of $(y_1, y_2, \ldots, y_q)$ such that any $y_j$ is not in ${\cal Y}$ is at most $q \cdot N e^{-\kappa \lp {\cal Y} \rp} = N e^{-\kappa \lp {\cal Y} \rp + \ln q} = N e^{- \Omega \lp k \rp}$ by \refeq{ComplOfCalY}. Hence the complement sparsity of ${\cal F}$ is also $\Omega \lp k \rp$.

In summary, we have constructed a family ${\cal F}$ of splits ${\bf y}$ from the given circuit $C$ computing {\sc CLIQUE}$_{n,k}$ satisfying: 
\begin{quote}
{\bf Property of y}: Each ${\bf y} \in {\cal F}$ is a split $(y_1, y_2, \ldots, y_q)$ of $[n]$ such that $y_j \in {\cal Y}$, $i.e.$,\\
\spca (i) $y_1 \cup y_2 \cup \cdots \cup y_q=[n]$,\\
\spca (ii) $y_j$ are pairwise disjoint, and\\
\spca (iii) every $y_j \in {[n] \choose n/q}$ of ${\bf y}$ is a valid set of a generator $g$ constructed at any\\
\spcb node $\alpha$; thus, there exists a $k$-clique $c$ generated at $\alpha$ such that\\
\spcb $g \subset c \subset g \cup y_j$.\\
In addition, the complement sparsity of ${\cal F}$ is $\Omega \lp k \rp$.
\end{quote}

\noindent
\refth{ExtensionGeneratorTheorem} allows for the construction of ${\cal F}$. Our alternative proof of the exponential monotone complexity of {\sc Clique}$_{n, k}$ will use a split ${\bf y} \in {\cal F}$.

\subsection{Extension Generator and Sunflower} \label{Sunflower}

Let $\Delta \in [n]$. A {\em $\Delta$-sunflower in $U \subseteq {[n] \choose m}$} consists of $m$-sets $s_1, s_2, \ldots, s_\Delta \in U$ such that $s_j = c \cup p_j$ for some pairwise disjoint $c$, $p_1$, $\ldots$, $p_q\subset [n]$ with $p_j \ne \emptyset$ \cite{jukna}. They are called the {\em core $c$} and {\em petals $p_j$} of the sunflower, respectively. The well-known sunflower lemma provides a size lower bound for $U$ that contains a $\Delta$-sunflower.

\begin{lemma} \label{ErdosRado} (Erd\"os, Rado)
Any $U \subseteq {[n] \choose m}$ whose size exceeds $\lp \Delta -1 \rp^m m!$ contains a $\Delta$-sunflower.
\end{lemma}

\medskip\medskip

Suppose $m=\sqrt[3] n$ and $\kappa \lp U \rp \le m \sqrt{\ln \ln n}$. In what follows, we construct with the extension generator theorem a $\Delta$-sunflower such that
\[
\Delta = \frac{n}{m^2 \ln^2 n}=\frac{n^{1/3}}{\ln^2 n},
\]
and its core has a small size $|c| \ll m$.

This is not easy to show by \reflm{ErdosRado} alone. For example, suppose that we have constructed a $\Delta/2$-sunflower $F$ of petal size $|p_j|=\Delta/2$ so far in an iterative application of the sunflower lemma. This situation is possible because the sunflowers directly produced by the lemma, denoted by $F'$,  could have small petal size $|p_j|=1$. A new $F'$ meaningful for the update of $F$ would have petals disjoint from those of $F$. We could need a $\Omega \lp m \Delta \rp$-sunflower $F'$ to guarantee this property. With $m \Delta =\Omega \lp n^{2/3} \Big/ \ln^2 n \rp$, it is difficult to have such $F'$ with a straightforward application of the sunflower lemma. Hence an extra observation is necessary.


Find a $\Delta$-sunflower $F$ in $U$ with \refth{ExtensionGeneratorTheorem} as follows. Put 
\[
l=\frac{n}{\Delta}=n^{2/3} \ln^2 n
\eqand
\lambda=2 \ln n.
\]
By the theorem, there exists an $\lp l, \lambda  \rp$-extension generator $g$ of $U$ such that
\[
|g| \le \frac{\kappa\lp U \rp}{\ln \frac{l}{m^2\lambda} - O(1)}  \le\frac{m \sqrt{\ln \ln n}}{\ln \frac{n^{2/3} \ln^2 n}{m^2 \ln n} - O(1)}
= \frac{m}{\Omega \lp \sqrt{\ln \ln n} \rp}
\ll m.
\]
Regard this $g$ as the core $c$ so its size $|c| \ll m$ is small. Construct a split $\lp y_1, y_2, \ldots, y_\Delta \rp$ of $[n]$ similar to ${\bf y}$ in \refsec{AppToClique}. Here $y_j$ are i) valid $\lp l- |g|\rp$-sets of $g$ in the space $[n] \setminus g$, ii) pairwise disjoint, and iii) such that each $g \cup y_j$ contains an $m$-set $s_j \in U$ since $g$ is an extension generator.

The split $\lp y_1, y_2, \ldots, y_\Delta \rp$ creates a $\Delta$-sunflower $F=\lb s_1, s_2, \ldots, s_\Delta \rb$ such that each $s_j$ in contained in $g \cup y_j$. Since $y_1, y_2, \ldots, y_\Delta$ are pairwise disjoint, $F$ is indeed a $\Delta$-sunflower with a small core $c=g$ and disjoint petals $p_j = s_j \setminus g \subset y_j$. 

We have a more general statement below:

\begin{proposition} \label{SmallCore} (Sunflower with a Small Core)
Let positive real numbers $m$ and $\eta$ satisfy $m \in [n]$, $\eta \gg 1$ and $m^2 \eta \ll \frac{n}{\ln n}$. A family of $m$-sets whose sparsity is $o \lp m \ln \eta \rp$ contains an $\Omega \lp \frac{n}{m^2 \eta \ln n} \rp$-sunflower with a core of size $o \lp m \rp$.
\end{proposition}
\begin{proof}
Take an $\lp l, 2 \ln n \rp$-extension generator of $U$ such that $l=m^2 \eta \ln n$, and follow the same arguments as above.
\end{proof}

\medskip

The proposition suggests a tight relationship between extension generators and sunflowers. \refth{ExtensionGeneratorTheorem} implies the existence of a sunflower in most cases covered by \reflm{ErdosRado}: One can verify that if $|U| = \lp \Omega \lp \Delta \ln \Delta \rp \rp^m m!$, there exists an $n / \Delta$-extension generator $g$ of size less than $m$. A $q$-sunflower is immediately constructible from such $g$ as above. This leads to non-trivial claims such as \refprop{SmallCore}. With \refth{ExtensionGeneratorTheorem}, we have a stronger statement that if we extend $m$-sets to $l$-sets, almost every case of disjoint $y_1, y_2, \ldots, y_\Delta$ with small $g$ includes a $\Delta$-sunflower.

\section{A New Approach to Show The Exponential Monotone Complexity of {\sc Clique} with Extension Generators} \label{MonotoneCircuitComplexity}

Despite no use of logical negations, it is a considerably large task to identify the monotone complexity of the clique function \refeq{CliqueFunction}. In this section, we demonstrate an alternative proof of its exponential lower bound with extension generators. The presented {\em shift method} dynamically constructs a counter example, a truth assignment whose positive literal set includes no $k$-clique. We believe that the shift method has a noteworthy difference from the standard Razborov-Alon-Boppana proof. We will discuss the difference and its possible applicability to a non-monotone circuit in the next section. 

We prove the following proposition with extension generators in the rest of this section:

\begin{proposition} \label{Main2}
A monotone circuit computing {\sc Clique}$_{n, \sqrt[4] n}$ has size 
$exp \lp \Omega \lp n^{\ep}  \rp \rp$ for a constant $\ep>0$. \qed
\end{proposition}

\subsection{Disjunctive Normal Form and Terms}

In the alternative approach, we deal with the disjunctive normal form (DNF) of the Boolean function $f_\alpha$ associated with a node $\alpha$ in a circuit. The DNF is in the heart of our effort to explore set theoretical properties of computation by Boolean circuits; therefore, it is represented as a family of sets of literals. Below we define this feature on a general circuit $C$.

A {\em De Morgan circuit} is a Boolean circuit consisting of the following four types of nodes \cite{jukna2}: i) conjunction (non-leaf, AND), ii) disjunction (non-leaf, OR), iii) a Boolean variable (leaf, {\em positive literal}) and iv) a negated Boolean variable (leaf, {\em negative literal}). A general Boolean circuit is converted into a De Morgan circuit of almost the same size, by pushing negations toward leaves by De Morgan's law. We consider a Boolean circuit $C$ of this form.

Denote by $\alpha$ any node of a De Morgan circuit $C$ as before, and by $r(C)$ the root of $C$. A node $\alpha_1$ is a {\em descendant} of $\alpha$ if there is a directed path from $\alpha$ to $\alpha_1$ in $C$ and $\alpha_1 \ne \alpha$. Since $C$ is directed and acyclic, we can align the nodes in a {\em topological order}, $i.e.$, a node order such that $\alpha$ is numbered before any descendant $\alpha_1$. The {\em depth of $\alpha$}, denoted by $depth \lp \alpha\rp$, is the maximum length of a directed path from $r \lp C \rp$ to $\alpha$.

Each node $\alpha$ of $C$ is associated with a Boolean expression $f_\alpha$ constructed inductively: If $\alpha$ is a leaf, $f_\alpha$ is logically equivalent to the literal associated with it. If $\alpha$ is a conjunction $\alpha_1 \wedge \alpha_2$, then $f_\alpha= \lp f_{\alpha_1}  \rp \wedge \lp f_{\alpha_2} \rp$. Otherwise it is a disjunction $\alpha =\alpha_1 \vee \alpha_2$, so $f_\alpha= f_{\alpha_1} \vee f_{\alpha_2}$. Regard $f_\alpha$ as a Boolean function also. The {\em Boolean function computed by $C$} is $f_{r(C)}$.

A leaf of $C$ to compute {\sc CLIQUE}$_{n, k}$ is associated with a literal: either $X_e$ or $\neg X_e$ for an edge $e \in {[n] \choose 2}$. As in \refeq{CliqueFunction}, we regard a 2-set $e \in {[n] \choose 2}$ as an edge in a graph whose vertex set is $[n]=\lb 1, 2, \ldots, n \rb$. For simplicity, write $(i, j)$ for $e= \lb i, j \rb \in {[n] \choose 2}$ and $X_e$, and $\neg (i, j)$ for $\neg X_e$. A {\em truth assignment $S$} is an input to $C$; it is a set of either $(i, j)$ or $\neg (i, j)$ but not both for every edge $(i, j) \in {[n] \choose 2}$. The circuit $C$ returns true or false to each $S$.


\begin{figure}
\centering
\includegraphics[width=100mm]{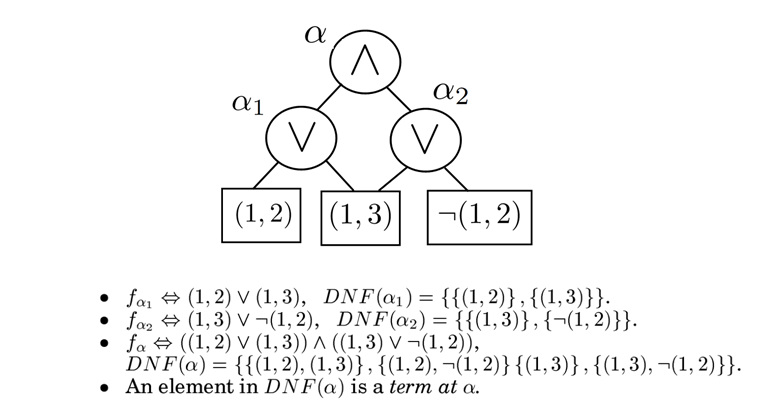}  
\caption{$DNF(\alpha)$ and Terms Included} \label{fig10}
\end{figure}

Define the {\em disjunctive normal form $DNF \lp \alpha \rp$ of $\alpha$} recursively:\\
i) If $\alpha$ is a leaf associated with a literal $(i, j)$ or $\neg (i, j)$, $DNF \lp \alpha \rp=\big\{ \lb (i,j) \rb  \big\}$ or $\big\{ \lb  \neg (i,j) \rb  \big\}$, respectively.\\
ii) If it is a conjunction $\alpha=\alpha_1 \wedge \alpha_2$, then $DNF \lp \alpha \rp= \lb t_1 \cup t_2 ~:~ t_i  \in DNF \lp \alpha_i \rp\rb$.\\
iii) Otherwise, it is a disjunction $\alpha=\alpha_1 \vee \alpha_2$. Define $DNF \lp \alpha \rp= DNF \lp\alpha_1 \rp \cup DNF \lp \alpha_2 \rp$.

$DNF \lp \alpha \rp$ is a family (set) of terms, where a {\em term} $t$ is a set of literals\footnote{A literal is a Boolean variable or its negation. A conjunction of literals is called {\em product term} in logic. When it is viewed in combinatorics as a product of variables in a polynomial over the binary finite field, it is called {\em monomial} \cite{jukna}. In this methodology, we especially emphasize its set theoretical side; a literal conjunction is thought of as a set of edges (positive literals) and non-edges (negative literals) over the vertices $[n]$. In this regard, we simply call it a term identifying it with the literal set.}
that means their conjunction. 
We say that $t$ is {\em at $\alpha$} if $t \in DNF \lp \alpha \rp$. It may contain a {\em contradiction} $\lb (i, j), \neg (i,j) \rb$ for an edge $(i, j) \in {[n] \choose 2}$. A term at the root $r(C)$ is said to be {\em global}. An example of $DNF(\alpha)$ in a De Morgan circuit is given in \reffig{fig10}.

\begin{remarks}
\item In the figure, we have the Boolean function $f_\alpha$ also regarded as a Boolean expression. Apply the distributive law
\[
x \vee \lp y \wedge z \rp =\lp x \wedge y \rp \vee \lp x \wedge z \rp,
\]
maximum times. Obtained is a disjunctive normal form of $f_\alpha$, and $DNF(\alpha)$ consists of all the terms occurring in it. This can be shown by induction on $depth(\alpha)$. 

\item A circuit $C$ returns true to an assignment $S$ if and only if $S$ contains a global term of $C$. The root $r(C)$ is logically equivalent to the disjunction of the global terms. If $C$ returns true to $S$, it means $S$ must imply at least one $t_0 \in DNF \lp r(C) \rp$, $i.e.$, $S \supseteq t_0$. If $C$ returns false to $S$, it includes no global term.
\item Such a global term $t_0$ is free from a contradiction. This is because $t_0 \subseteq S$  for a truth assignment $S$ that is contradiction-free by definition. 
\item When $C$ computes the clique function {\sc CLIQUE}$_{n, k}$, every contradiction-free global term $t_0$ contains a $k$-clique in its positive literal set. Otherwise, we could find  $S \supseteq t_0$ by adding negative literal $\neg e$ for any edge $e$ missing in $t_0$. The circuit $C$ returns true to this $S$ but it contains no $l$-clique. It is against $r(C) \Leftrightarrow$ {\sc CLIQUE}$_{n, k}$. Thus each $t_0$ of $C$ contains a $k$-clique.

\end{remarks}

Here is our general scheme to show the impossibility of $C$ to compute {\sc CLIQUE}$_{n, k}$: We systematically construct a global term $t_0$ that is free from both a $k$-clique and contradiction. Such $t_0$ is called a {\em shift}. In this paper, $C$ is monotone so a shift is a global term containing no $k$-clique.

\pdffigure{fig6}{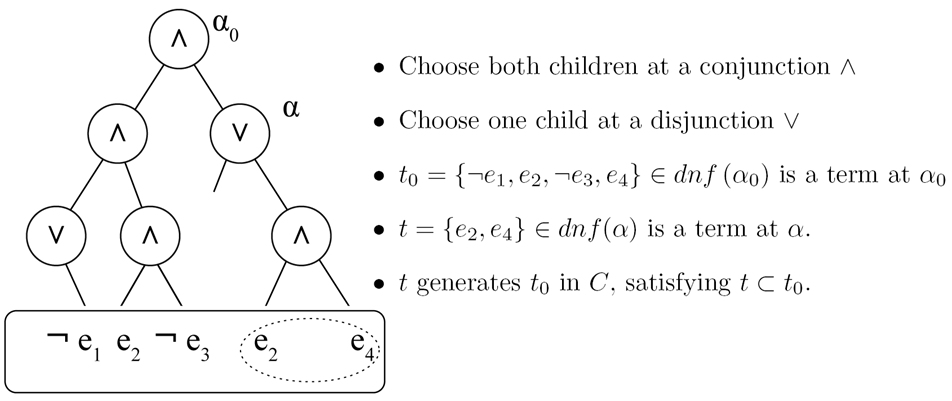}{Derivation Graph of $t_0$ Where $t$ Generates $t_0$ in $C$}

A $k$-clique $c$ is also a $k$-set in the universal space $[n]$ so we may write $c \in {[n] \choose k}$. In \refsec{AppToClique}, we have defined that $c$ is generated at $\alpha$ if its edges ${c \choose 2}$ imply $\alpha$. Equivalently, $c$ is generated at $\alpha$ iff there exists a term $t \subseteq {c \choose 2}$ at $\alpha$.

For a given term $t$ at $\alpha$, construct a subgraph $C'$ of $C$ by the following process: Suppose that $\alpha$ is a conjunction $\alpha_1 \wedge \alpha_2$. By the definition of $t \in DNF \lp \alpha \rp$, there exists a term $t_i$ at $\alpha_i$ ($i=1,2$) such that $t= t_1 \cup t_2$. Then we visit both $\alpha_i$ recursively. If $\alpha$ is a disjunction $\alpha_1 \vee \alpha_2$, the term $t$ is included in $DNF(\alpha_i)$ with $i=1$ or $2$. Visit this $\alpha_i$ recursively.

The subgraph $C'$ of $C$, the set of nodes and edges traversed by the process, is called a {\em derivation graph of $t$}. Observe that it constructs a term $t'$ at every visited node such that $t' \subseteq t$. We say that the term {\em $t'$ generates $t$ in $C$.}

Notice that the process can visit a node $\alpha'$ more than once, creating a possibly different term at $\alpha'$ per visit. We can reconstruct $t$ into a {\em minimal term} by choosing a same term $t'$ at any visited $\alpha'$. The obtained term is a subset of $t$ that belongs to $DNF \lp \alpha \rp$ also. By our convention, a global term $t_0$ means a minimal term at the root $r \lp C \rp$ with a specified derivation graph. There uniquely exists $t$ at every node $\alpha$ in $C$ that generates $t_0$ in $C$, or no such $t$ exists at $\alpha$.

\subsection{Preprocess with Extension Generators}

To prove \refprop{Main2}, we falsely assume a monotone circuit $C$ of size at most $\exp \lp n^{\ep^2} \rp$ to compute {\sc CLIQUE}$_{n,k}$ where
$
k= \sqrt[4] n
$.
Below we describe a preprocess for our shift construction.

\subsubsection{Finding Generators for $k$-Cliques}

As its first step, we find extension generators $g$ of generated cliques $c$ at every node $\alpha$. The algorithm given below generalizes the construction described in \refsec{AppToClique}. It holds a small set of cliques excluded from consideration, the set of {\em error cliques} denoted by ${\cal C}_\alpha \subset {[n] \choose k}$.

Arrange all $\alpha$ in a reverse topological order, and perform the algorithm at each $\alpha$. If $\alpha$ is a leaf, let ${\cal C}_\alpha=\emptyset$. Otherwise ${\cal C}_\alpha={\cal C}_{\alpha_1} \cup {\cal C}_{\alpha_2}$ where $\alpha_i$ are the two children of $\alpha$. Find generators $g$ by:

\begin{quote}
{\bf Algorithm {\sc CliqueGenerators}}\\
Collect all the $k$-cliques $c \not\in {\cal C}_\alpha$ generated at $\alpha$. Regarding $c$ as $k$-sets in $[n]$, find their $\lp n/q, k \rp$-extension generator $g$. (We defined $q=n^{5\ep}$ and $k=n^{1/4}$ in \refeq{PreliminaryNotation}.) Exclude all $c$ such that $c \supset g$. Repeat until the number of remaining cliques $c$ is less than ${n \choose l}e^{-n^\ep}$. When the loop is finished, add the finally remaining $k$-cliques to ${\cal C}_\alpha$ as the new error cliques at $\alpha$.
\end{quote}

We have
\beeq{SizeBoundsOfGandC}
|g| = O \lp \frac{n^\ep}{\ln n} \rp
\eqand
\kappa \lp {\cal C}_\alpha \rp = \Omega \lp n^\ep \rp.
\eeq
When we take $g$ of $c$ at $\alpha$, there are at least ${n \choose l}e^{-n^\ep}$ $k$-cliques. The sparsity of (the family of $k$-sets) $c$ is $n^\ep$ or less. The size of $g$ is $O \lp \frac{n^\ep}{\ln n} \rp$ by the extension generator theorem. Since less than ${n \choose l}e^{-n^\ep}$ cliques are added to ${\cal C}_\alpha$ at each $\alpha$, the final size of ${\cal C}_\alpha$ is bounded by $|C| \cdot {n \choose l}e^{-n^\ep}=e^{n^{\ep^2}}\cdot {n \choose l}e^{-n^\ep}={n \choose l}e^{-\Omega \lp n^\ep \rp}$. These mean \refeq{SizeBoundsOfGandC}, $i.e.$, generators $g$ have small sizes and error cliques are sufficiently few.

\subsubsection{Splits ${\bf y}$ of $[n]$} \label{SplitYofN}
In \refsec{AppToClique}, we constructed the family ${\cal Y}$ of $n/q$-sets valid for $g$ at all $\alpha$, and ${\cal F}$ of splits ${\bf y}=\lp y_1, y_2, \ldots, y_q \rp$ of $[n]$ such that $y_j \in {\cal Y}$. Here we construct similar ${\cal Y}$ and ${\cal F}$ for possibly more than one generators at a node $\alpha$. Use the same notation as \refeq{PreliminaryNotation}.

Consider one $\lp n/q, k \rp$-extension generator $g$ found by {\sc CliqueGenerators} at a particular $\alpha$. By \refco{ExtensionGeneratorTheorem2}, there are ${n \choose n/q}e^{-k}$ valid sets $y$, for each of which $g \cup y$ contains a $k$-clique generated at $\alpha$.

Let ${\cal Y} \subseteq {[n] \choose n/q}$ be the family of $n/q$-sets valid for all $g$. By \reflm{asymptotic} and \refeq{SizeBoundsOfGandC}, there are at most ${n \choose O \lp \frac{n^\ep}{\ln n} \rp} |C|=e^{O \lp n^\ep \rp}$ pairs of generators $g$ and node $\alpha$ where $g$ is found. There are no more than 
\[
{n \choose n/q}e^{-k}\cdot e^{O \lp n^\ep \rp} 
={n \choose n/q}e^{-\Omega \lp k \rp}
\]
$n/q$-sets $y$ not valid for any detected $g$. Thus the complement sparsity of ${\cal Y}$ is $\Omega \lp k \rp$, $i.e.$, \refeq{ComplOfCalY} holds for more than one $g$ at an $\alpha$.

Now construct the family ${\cal F}$ of splits ${\bf y}= \lp y_1, y_2, \ldots, y_q \rp$ such that $y_j \in {\cal F}$ the same way as in \refsec{AppToClique}. It satisfies the properties (i)--(iii) and $\kappa \lp \overline{\cal F}\rp = \Omega \lp k \rp$. Also a $k$-clique $c$ contained in each $g \cup y_j$ is not an error clique at $\alpha$, $i.e.$, $c \in \overline{{\cal C}_\alpha}={[n] \choose k} \setminus {\cal C}_\alpha$. It is due to the construction by {\sc CliqueGenerators}.

\subsubsection{Quadruples $\sigma$: Representation of Generated Cliques}

Let $\sigma =\lp g, g_1, g_2, \alpha \rp$ be a quadruple such that $g$,  $g_1$ and $g_2$ are subsets of $[n]$ and $\alpha$ is a node in $C$. It is said to be {\em incident to a $k$-clique $c$ at $\alpha$} if it satisfies the following four conditions:
\begin{enumerate}[I.]
\item $c$ a non-error $k$-clique is generated at $\alpha$ to contain some generator $g$ found by {\sc CliqueGenerators} at $\alpha$ (so $c \in \overline{{\cal C}_\alpha}$ and $c \supset g$).
\item If $\alpha$ is a conjunction $\alpha_1 \wedge \alpha_2$, each $g_i$ is a generator found by {\sc CliqueGenerators} at $\alpha_i$ such that $g_i \subset c$.
\item If $\alpha$ is a disjunction $\alpha_1 \vee \alpha_2$, then $g_1=g_2$, which is a generator found by {\sc CliqueGenerators} at either $\alpha_i$ where $c \supset g_i$ is generated.
\item If $\alpha$ is a leaf of $C$, $g=g_1=g_2$ is the vertex set of the positive literal (edge) associated with $\alpha$.
\end{enumerate}

Denote by ${\cal Q}_0$ the set of all $\sigma$ incident to any generated $c$ at $\alpha$. A quadruple $\sigma$ incident to $c$ is an abstraction of $k$-clique generated at $\alpha$. All the operations in our shift construction will perform with $\sigma$. Observe their basic properties. 

\begin{lemma} \label{PropertyOfQ0}
\begin{enumerate}[(i)]
\item $\left| {\cal Q}_0 \right|< e^{O \lp n^{\ep} \rp} \ll e^{q}$.
\item For each $k$-clique $c \in \overline{{\cal C}_\alpha}$ generated at $\alpha$, and $g \subset c$ found by {\sc CliqueGenerators} at $\alpha$, there exists $\sigma=\lp g, g_1, g_2, \alpha \rp \in {\cal Q}_0$ incident to $c$.
\end{enumerate}
\end{lemma}
\begin{proof}
(i): Since $|g|$ and $|g_i|$ are bounded as \refeq{SizeBoundsOfGandC} and $|C| \le e^{n^{\ep^2}}$, there are at most ${n \choose  O \lp n^\ep / \ln n \rp}^3  e^{n^{\ep^2}}=e^{O \lp n^{\ep} \rp}$ such $\sigma=\lp g, g_1, g_2, \alpha \rp \in {\cal Q}_0$ by \reflm{asymptotic}. We defined $q=n^{5\ep}$ in \refsec{AppToClique}, so $e^{O \lp n^{\ep} \rp} \ll e^{q}$.

\medskip

(ii): Fix any such $c$ and $g$. There exists a term $t \in DNF \lp \alpha \rp$ contained in ${c \choose 2}$ since $c$ is generated at $\alpha$. Suppose $\alpha$ is a conjunction $\alpha_1 \wedge \alpha_2$. The set $DNF \lp \alpha_i \rp, i=1,2$ contains a term that is a subset of $t$, by the definition of $DNF \lp \cdot \rp$. The clique $c$ is generated at both $\alpha_i$. Algorithm {\sc CliqueGenerators} finds $g_i \subset c$ at each $\alpha_i$ since $c \not \in {\cal C}_\alpha \supseteq {\cal C}_{\alpha_i}$. Thus there exists such a quadruple $\lp g, g_1, g_2, \alpha \rp \in {\cal Q}_0$. It is shown similarly when $\alpha$ is a disjunction or a leaf. 
\end{proof}

\subsection{Overview of the Shift Construction for Monotone $C$}

For each $\sigma =\lp g, g_1, g_2, \alpha \rp \in {\cal Q}_0$, we will construct a term $t\lp \sigma \rp$ at $\alpha$ such that
\beeq{DisjointnessProperty}
t \lp \sigma \rp \setminus {g \choose 2} \cap z= \emptyset.
\eeq
Here $z$ is an edge set of size $n^{11/6} q=n^{11/6 + 5 \ep}$ whose removal annihilates $k$-cliques in the edge space ${[n] \choose 2}$, called {\em blocked edge set}.

Our shift, denoted by $t \lp {\bf y}\rp$ for a given split ${\bf y}$ of $[n]$ in ${\cal F}$, is any $t\lp \sigma \rp$ at the root $\alpha =r(C)$. It satisfies
\beeq{MonotoneGoal}
t \lp {\bf y}\rp \cap z = \emptyset:
\eeq
All the $k$-clique are generated at $r(C)$ so $g=\emptyset$. Then \refeq{DisjointnessProperty} means \refeq{MonotoneGoal}. Since $z$ is chosen so that its removal leaves no $k$-clique, the shift is free from a $k$-clique. Its existence is the contradiction to prove \refprop{Main2}.

We will show \refeq{DisjointnessProperty} by induction on $depth(\alpha)$. Suppose $\alpha = \alpha_1\wedge \alpha_2$. Inductively, two terms $t\lp \sigma_i \rp, i=1,2$ at the children $\alpha_i$ have been already constructed with the property \refeq{DisjointnessProperty}. Precisely speaking, let $\sigma$ be incident to a $k$-clique $c$ at $\alpha$. There exists $t \lp \sigma_i \rp \in DNF \lp \alpha_i \rp$ for each $i=1,2$ and some $\sigma_i =\lp g_i, g_{i, 1}, g_{i,2}, \alpha_i \rp$ incident to $c$ such that $t \lp \sigma_i \rp \setminus {g_i \choose 2} \cap z= \emptyset$. (This {\em chain property} will be later discussed with \reflm{PropertyOfQ0} (ii).)

We join the two terms $t(\sigma_i)$ to construct $t \lp \sigma \rp$ if $\alpha$ is a conjunction $\alpha_1 \wedge \alpha_2$. It is a term at $\alpha$ since 
\[
t \lp \alpha \rp = t\lp \alpha_1 \rp \cup t\lp \alpha_2 \rp \in DNF(\alpha).
\]
by the definition of $DNF(\alpha)$. To show \refeq{DisjointnessProperty} for $t \lp \sigma \rp$, it is critical whether or not
\[
{g_1 \cup g_2 \choose 2} \setminus {g \choose 2}  \cap z= \emptyset,
\]
$i.e.$, we already know the the disjointness between $t\lp \sigma \rp$ and $z$ outside the space $g_1 \cup g_2$. Put
\beeq{DefOfDsigma}
d\lp \sigma \rp \stackrel{def}{=} {g_1 \cup g_2 \choose 2} \setminus {g \choose 2}.
\eeq
to restate it as
\beeq{PreliminaryBasicProperty}
d \lp \sigma \rp \cap z= \emptyset.
\eeq

Fix a split ${\bf y}= \lp y_1, y_2, \ldots, y_q \rp \in {\cal F}$ obtained in \refsec{SplitYofN}. For any $g$ at $\alpha$, each $g \cup y_j$ contains a $k$-clique generated at $\alpha$. We have $q$ choices $y_1, y_2, \ldots, y_q$ for a generated clique. An algorithm will choose one of them $y_j$ for each $\sigma \in {\cal Q}_0$. Write
\[
y \lp \sigma \rp = y_j
\]
to exress it. We achieve \refeq{PreliminaryBasicProperty} by determining a good choice $y \lp \sigma \rp \in \lb y_1, y_2, \ldots, y_q \rb$ for every $\sigma \in {\cal Q}_0$. The quadruple $\sigma= \lp g, g_1, g_2, \alpha \rp$ is incident to some $c \in \overline{{\cal C}_\alpha}$. There exists a $k$-clique 
\[
c^* \subset g \cup y \lp \sigma \rp.
\]
The choice $y \lp \sigma \rp$ represents our systematic switch from $c$ to $c^*$ to construct $t \lp {\bf y }\rp$.

By \reflm{PropertyOfQ0} (ii), there exists $\sigma^*= \lp g, g^*_1, g^*_2, \alpha \rp$ incident to $c^*$. This quadruple occurs {\em after} making the decision $y \lp \sigma \rp \in \lb y_1, y_2, \ldots, y_q \rb$, thus is our object for verifying the property \refeq{PreliminaryBasicProperty}. We require $\sigma^*$ to satisfy
$
d \lp \sigma^* \rp \cap z= \emptyset
$.
The mapping $\sigma \mapsto \sigma^*$ is the primary decision to switch from clique $c$ to $c^*$. Let us denote it by $\sigma^*=f \lp \sigma \rp$ with the function $f:{\cal Q}_0 \rightarrow {\cal Q}_0$. As a result, we require
\beeq{BasicProperty}
d \lp f \lp \sigma \rp \rp \cap z=\emptyset.
\eeq
for every $\sigma \in {\cal Q}_0$.

\medskip\medskip

In summary, our construction will determine:
\begin{enumerate} [(a)]
\item a blocked edge set $z$ of size $n^{11/6}q$ whose removal leaves no $n^{1/5}$-cliques in the edge space ${[n] \choose 2}$, and
\item $y\lp \sigma \rp \in \lb y_1, y_2, \ldots, y_q \rb$ and $f\lp \sigma \rp \in {\cal Q}_0$ such that \refeq{BasicProperty} for every $\sigma \in {\cal Q}_0$.
\end{enumerate}
Based on the choice, we will find $t\lp \sigma \rp$ such that \refeq{DisjointnessProperty} and $t \lp {\bf y} \rp$. \refprop{Main2} will be proven by the confirmed existence of a shift.

\subsection{The Main Algorithm}

Our main algorithm {\sc Shift} returns $t \lp {\bf y }\rp$ of $C$ for each ${\bf y}=\lp y_1, y_2, \ldots, y_q \rp \in {\cal F}$. It consists of the following 4 steps:
\beenu
\item Initialization:
\beenub
\itembb Put ${\cal Q}={\cal Q}_0$.\\
/* It is the set of all $\sigma = \lp g, g_1, g_2, \alpha \rp$ incident to any $c$ at $\alpha$. */
\itemb For every $\sigma = \lp g, g_1, g_2, \alpha \rp \in {\cal Q}_0$ and $y_j \in \lb y_1, y_2, \ldots, y_ q\rb$, choose a $k$-clique $c^*\in \overline{{\cal C}_\alpha}$ such that $ c^*\subset g \cup y_j$. Let $\sigma^* \in {\cal Q}_0$ be incident to $c^*$. Put
\[
f_j \lp \sigma \rp= \sigma^*.
\]
\eenub
/* By \reflm{PropertyOfQ0}, there exists such a $\sigma^*$ incident to $c^*$. This $f_j \lp \sigma \rp$ is $f \lp \sigma \rp$ in case $y_j$ is chosen as $y\lp \sigma \rp$.*/
\item Call Algorithm {\sc BlockedEdges} (\reffig{BlockedEdges}) to determine $z$, $y \lp\sigma \rp$ and $f\lp \sigma \rp$ for $\sigma \in {\cal Q}_0$.
\item Call Algorithm {\sc LocalShift} (\reffig{LocalShift}) to construct a term $t \lp \sigma \rp$ for each $\sigma \in {\cal Q}_0$. 
\item Return $t \lp \sigma \rp$ for any $\sigma = \lp \emptyset, g_1, g_2, r(C) \rp$ as $t \lp {\bf y}\rp$.
\eenu

\medskip

\subsection{Algorithm {\sc BlockedEdges}}

Fix a given split ${\bf y}=\lp y_1, y_2, \ldots, y_q \rp  \in {\cal F}$ of $[n]$. Step 2 of {\sc Shift} calls the algorithm {\sc BlockedEdges} to determine a blocked edge set $z$, and choices $y \lp \sigma \rp$ and $f(\sigma)$ for $\sigma \in {\cal Q}_0$. It is described in \reffig{BlockedEdges}.

We choose an edge set $z_j$ of size $n^{11/6}$ in the space $y_j$ for each $j \in [q]$. Our notation expresses it as
\[
z_j \in {{y_j \choose 2} \choose n^{11/6}}.
\]
Step 3 constructs the blocked edge set $z$ by
\[
z = \bigcup_{j=1}^q z_j \in {{[n] \choose 2} \choose n^{11/6}q}.
\]

We require for each $j \in [q]$ that the removal of $z_j$ leave no $n^{1/5}$-clique in the space $y_j$. There are a majority of such $z_j$. See it with the following general statement.

\begin{lemma} \label{CliquelessEdges}
Let $N_0={n \choose 2}$, $r\in [n]$ and $L \in \lbr N_0 \rbr$ be integers such that $r \ll n$ and $r L \gg N_0 \ln \frac{\sqrt{N_0}}{r}$. The family
$
\lb s \in {{[n] \choose 2} \choose N_0-L} ~:~ \textrm{edge set $s$ contains no $r$-clique} \rb
$
forms a majority in ${{[n] \choose 2} \choose N_0-L}$ with complement sparsity $\Theta \lp L r^2 \big/ N_0 \rp$.
\end{lemma}
\begin{proof}
Let $R ={r \choose 2}$. The number of $\lp N_0 - L \rp$-edge sets containing a particular $r$-clique is
${N_0 - R \choose N_0 - L - R}={N_0 - R \choose L} = {N_0 \choose L}e^{-\Theta \lp \frac{LR}{N_0} \rp}={N_0 \choose N_0-L}e^{- \Theta \lp \frac{LR}{N_0} \rp}$ by \reflm{basic2}. The number of $\lp N_0 - L \rp$-edge sets containing any $r$-clique is
${N_0 \choose N_0-L}e^{- \Theta \lp  \frac{LR}{N_0} \rp} \cdot {n \choose r}={N_0 \choose N_0-L}e^{- \Theta \lp \frac{LR}{N_0} \rp + \ln {n \choose r}}$. The sparsity of such $\lp N_0 - L \rp$-edge sets is 
\[
\Theta \lp \frac{LR}{N_0} \rp - \ln {n \choose r} = \Theta \lp \frac{L r^2}{N_0} \rp - \Theta \lp r \ln \frac{n}{r} \rp 
= \Theta \lp \frac{L r^2}{N_0} \rp 
\gg \Theta \lp r \ln \frac{n}{r} \rp \gg 1, 
\]
due to \reflm{asymptotic} and the given conditions $rL  \gg N' \ln \frac{\sqrt{N_0}}{r}$ and $\frac{n}{r} \gg 1$. Hence the $\lp N_0 - L \rp$-edge sets containing no $r$-cliques form a majority with complement sparsity $\Theta \lp L r^2 \big/ {N_0} \rp$.
 \end{proof}

By the lemma, the family
\[
{\cal Z}_j \stackrel{def}{=} \lb z_j \in {{y_j \choose 2} \choose n^{11/6}}~:~ \textrm{removal of $z_j$ leaves no $n^{1/5}$-cliques in the space $y_j$}~\rb
\]
forms a majority in ${{y_j \choose 2} \choose n^{11/6}}$. Put $r \leftarrow n^{1/5}$, $L \leftarrow \left| z_j \right|=n^{11/6}$, and $N_0 \leftarrow {|y_j| \choose 2}={n/q \choose 2}=\Theta\lp n^{2-10\ep} \rp$. Apply them to the lemma. Since 
\[
\frac{r L}{N_0 \ln \frac{\sqrt N_0}{r}} = \frac{n^{1/5} \cdot n^{11/6}}{\Theta \lp n^{2-10\ep} \ln n \rp} \gg 1
~~~\Rightarrow~~~
r L \gg N_0 \ln \frac{\sqrt N_0}{r}, 
\]
the family of $\lp N_0-L \rp$-edge sets containing no $r$-clique forms a majority with complement sparsity $\Omega \lp L r^2 \big/ N_0 \rp = \Omega \lp n^{7/30} \rp$. It means that ${\cal Z}_j$ forms a majority with the same complement sparsity in ${{y_j \choose 2} \choose n^{11/6}}$.

\begin{figure} 
\medskip
\begin{small}
{\bf Algorithm {\sc BlockedEdges}}\\
{\bf Inputs:} \\
\spca 1. Split ${\bf y}=\lp y_1, y_2, \ldots, y_ q \rp$ of $[n]$ input to {\sc Shift}.\\
\spca 2. Set ${\cal Q}={\cal Q}_0$ of all $\sigma = \lp g, g_1, g_2, \alpha \rp$ constructed by Step 1-1 of {\sc Shift}.\\
\spca 3. $f_j \lp \sigma \rp \in {\cal Q}_0$ constructed by Step 1-2 of {\sc Shift} for $\sigma \in {\cal Q}_0$ and $j \in [q]$.\\
{\bf Outputs:} (i) blocked edge set $z$ of size $n^{11/6}q$, (ii) $y\lp \sigma \rp \in \lb y_1, y_2, \ldots, y_q \rb$, and (iii) $f\lp \sigma \rp \in {\cal Q}_0$ such that \refeq{BasicProperty} for every $\sigma \in {\cal Q}_0$.

\noindent
\bbegin
\beenu
\item \bfor $j \leftarrow 1$ \bto $q$ \bdo
	\beenub
	\itembb \bfor each $\sigma \in {\cal Q}$ \bdo $y \lp \sigma \rp\leftarrow y_j$ and $f \lp \sigma \rp \leftarrow f_j \lp \sigma \rp$;\\
	/* $y_j$ is chosen as $y \lp \sigma \rp$ and $f_j $ as $f \lp \sigma \rp$ temporarily for the remaining $\sigma$ in the current ${\cal Q}$. */
	\itemb \bfor each $\sigma \in {\cal Q}$ and edge set $z_j \subset {y_j \choose 2}$ of size $|z_j|=n^{11/6}$ \bdo
		\beenuc
		\itemcc \bif i) removal of $z_j$ leaves no $n^{1/5}$-cliques in ${y_j \choose 2}$, and ii) $z_j \cap d\lp f(\sigma) \rp \ne \emptyset$ \bthen create a pair $\lp z_j, \sigma \rp$;
		\eenuc
	\itemb \bend \bfor
	\itemb Find and fix $z_j$ incident to the minimum number of $\lp z_j, \sigma  \rp$;
	\itemb ${\cal Q}(z_j) \leftarrow \lb \sigma \in {\cal Q} ~:~ d\lp f(\sigma) \rp \cap z_j \ne \emptyset \rb$;
	\itemb ${\cal Q}_j \leftarrow {\cal Q} \setminus {\cal Q}(z_j)$ and ${\cal Q}\leftarrow {\cal Q}(z_j)$;
	\eenub
\item \bend \bfor
\item $z=z_1 \cup z_2 \cup \cdots \cup z_q$;
\eenu
\bend
\end{small}
\caption{To Determine $z$, $y \lp \sigma \rp$ and $f \lp \sigma \rp$} \label{BlockedEdges}
\end{figure}

\medskip \medskip

Let us look at details in {\sc BlockedEdges}. Loop 1 processes $y_1, y_2, \ldots, y_j, \ldots, y_q$ in the order. In its $j^{th}$ step, the set ${\cal Q}$ consists of all $\sigma$ for which $y \lp \sigma \rp$ and $f \lp \sigma \rp$ are yet to be determined. Step 1-1 temporarily sets $y \lp \sigma \rp=y_j$ and $f \lp \sigma \rp=f_j \lp \sigma \rp$ for the remaining $\sigma \in {\cal Q}$.

We require that $z_j$, $y\lp \sigma \rp$ and $f \lp \sigma \rp$ satisfy 
\[
(A): z_j \in {\cal Z}_j
\eqand
(B): d \lp f(\sigma)\rp \cap z_j=\emptyset,
\]
to meet \refeq{BasicProperty}. Steps 1-2 to 1-4 construct the pairs $(z_j, \sigma)$ such that $(A) \wedge \neg (B)$, finding  $z_j \in {\cal Z}_j$ incident to the fewest $\sigma$. Fix this $z_j$ as $z \cap {y_j \choose 2}$. Step 1-5 stores them in ${\cal Q}(z_j)$. For the remaining ${\cal Q}\setminus {\cal Q}\lp z_j \rp$, the decisions are made as $y \lp \sigma \rp=y_j$ and $f \lp \sigma \rp=f_j \lp \sigma \rp$. Step 1-6 puts
\beeq{MonotoneQj}
{\cal Q}_j = \lb \sigma \in {\cal Q}_0 ~:~ y \lp \sigma \rp= y_j \rb.
\eeq
Also it updates ${\cal Q} \leftarrow {\cal Q} \lp z_j \rp$, $i.e.$, for $\sigma$ violating the condition $(B)$, the algorithm chooses another $y\lp \sigma \rp \in \lb y_{j+1}, y_{j+2}, \ldots, y_q \rb$. Due to the disjointness between $y_1, y_2, \ldots, y_q$, the construction of $z_{j+1}, z_{j+2}, \ldots, z_q$ performs independently of $j^{th}$ step. 

The three noteworthy features of the shift construction are:\\
I. {\em the disjointness property} \refeq{DisjointnessProperty} based on \refeq{BasicProperty}, \\
II. {\em the independence property between $y_1, y_2, \ldots, y_q$} as above, and \\
III. {\em the convergence property of ${\cal Q}_j$}, $i.e.$, $|{\cal Q}_j|$ decreases exponentially as $j$ grows linearly.

\medskip

\noindent
We will discuss I and II in the next subsection. We show III by proving the following lemma.

\medskip

\begin{lemma} \label{Converge}
$\left| {\cal Q}(z_j) \right|  \ll \left| {\cal Q} \right|$ right after Step {\em 1-5} of {\sc BlockedEdges} for $j \in [q]$.
\end{lemma}
\begin{proof}
Let $e$ be an edge in $z_j \cap d\lp f(\sigma) \rp$ for Step 1-2-1 (ii) of {\sc BlockedEdges}. It creates at most ${{|y_j| \choose 2}-1 \choose |z_j|-1}$ pairs $\lp z_j, \sigma \rp$; choose $|z_j|-1$ edges of $z_j$ in the space ${y_j \choose 2} \setminus \lb e \rb$. The number of such $e$ per $\sigma \in {\cal Q}$ is no more than
$
\left| d \lp f(\sigma) \rp \right| \le {\left| g^*_1 \cup g^*_2 \right| \choose 2}< n^{2\ep}
$
since we know $|g^*_i| = O \lp \frac{n^\ep}{\ln n} \rp<\frac{n^\ep}{2}$ by \refeq{SizeBoundsOfGandC}. We have the identity 
$
{{|y_j| \choose 2} \choose |z_j|} = \frac{{|y_j| \choose 2}}{|z_j|} {{|y_j| \choose 2}-1 \choose |z_j|-1}
$.
As a result, there are at most
\[
M={{|y_j| \choose 2}-1 \choose |z_j|-1} \left| d \lp f(\sigma) \rp \right|
< \frac{|z_j| n^{2 \ep}}{{|y_j| \choose 2}} {{|y_j| \choose 2} \choose |z_j|}
\]
such pairs $\lp z_j, \sigma \rp$ incident to each $\sigma \in {\cal Q}$. The total number of $\lp z_j, \sigma \rp$ constructed by Step 1 is $M \left| {\cal Q} \right|$ or less.

Therefore, there exists an edge set $z_j \in {\cal Z}_j$ incident to
\[
\frac{M \left| {\cal Q} \right|}{|{\cal Z}_j|}
=
\frac{M \left| {\cal Q} \right|}{{{|y_j| \choose 2} \choose |z_j|} e^{-o(1)}}
< \frac{|z_j|n^{2\ep}}{{|y_j| \choose 2} e^{-o(1)}} \left| {\cal Q} \right|
= O \lp n^{-\frac{1}{6} + 12 \ep}  \left| {\cal Q} \right| \rp
\ll \left| {\cal Q} \right|,
\]
or less $\sigma \in {\cal Q}$ in the set of constructed pairs $\lp z_j, \sigma \rp$.  Step 1-5 chooses these $\sigma$ as ${\cal Q}(z_j)$ so $\left| {\cal Q}(z_j) \right| \ll \left| {\cal Q} \right|$. 
\end{proof}

\begin{corollary} \label{SizeOfQj}
\[
\left| {\cal Q}_j \right| < 2^{-j+1} \left| {\cal Q}_0 \right|,
\]
for each $j \in [q]$ after {\sc BlockedEdges} terminates. As a result, ${\cal Q}_q$ is empty. 
\end{corollary}
\begin{proof}
Every time Loop 1 of {\sc BlockedEdges} increments $j$, it reduces $\left| {\cal Q} \right|$ to less than its half by \reflm{Converge}. Thus $\left| {\cal Q} \right| < 2^{-j+1} \left| {\cal Q}_0 \right|$ after Step 1-6 with loop index $j$. The first statement follows this and ${\cal Q}_j \subseteq {\cal Q}$. 

$|{\cal Q}_0|=e^{O\lp n^{\ep}\rp}$ by \reflm{PropertyOfQ0} (i). For every $j \ge n^{2\ep}$, 
\(
&&
|{\cal Q}_j| < 2^{-j+1} |{\cal Q}_0| < e^{-n^{2\ep}+1 + O \lp n^{\ep }\rp}=
e^{-\Omega \lp n^{2\ep }\rp}<1,
\sothat
{\cal Q_j} = \emptyset.
\)
So ${\cal Q}_q=\emptyset$ since $q=n^{5\ep}$.
\end{proof}

Hence ${\cal Q}_j$ quickly reduces its size to becomes empty before $j$ reaches $q$. This is the convergence property of the shift construction.

\subsection{Algorithm {\sc LocalShift}}

\begin{figure} 
\begin{small}
{\bf Algorithm {\sc LocalShift}}\\
{\bf Inputs:} \\
\spca 1. ${\cal Q}_0$ constructed by Step 1-1 of {\sc Shift}.\\
\spca 2. $f \lp \sigma \rp$ determined by {\sc BlockedEdges}.\\
{\bf Output:} A term $t \lp \sigma \rp$ for every $\sigma \in {\cal Q}_0$.

\noindent
\bbegin\\
/* Do the following for each $\sigma=\lp g, g_1, g_2, \alpha \rp \in {\cal Q}_0$ in a reverse topological order of the nodes in $C$. Inductively, we have constructed a term $t \lp \sigma_i \rp$ for every $\sigma_i \in {\cal Q}_0$ incident to a $k$-clique at a child of $\alpha$. */
\beenu
\item \bif $\alpha$ is a leaf of $C$ associated with a positive literal $e$ ($i.e.$, edge $e \in {[n] \choose 2}$) \bthen $t \lp \sigma \rp=\lb e \rb$;
\item \belse 
	\beenub
	\itembb $\sigma^*=\lp g, g_1^*, g_2^*, \alpha \rp \leftarrow f \lp \sigma \rp$;
	\itemb \bif $\alpha$ is a conjunction $\alpha_1 \wedge \alpha_2$ \bthen
		\beenuc
		\itemcc \bfor $i \leftarrow 1,2$ \bdo $\sigma_i \leftarrow \lp g^*_i, g_{i, 1}, g_{i, 2}, \alpha_i \rp \in {\cal Q}_0$ for some $g_{i,1}, g_{i,2} \subset [n]$; \\
		/* There exists such $\sigma_i$ by the chain property. */
		\itemc \breturn $t\lp \sigma_1 \rp \cup t \lp \sigma_2 \rp$;
		\eenuc
	\itemb \belse ~ /* $\alpha$ is a disjunction $\alpha_1 \vee \alpha_2$ */
		\beenuc
		\itemcc Let $g^*_1=g^*_2$ be a generator constructed by {\sc CliqueGenerators} at $\alpha_i$ for $i=1$ or $i=2$;
		\itemc $\sigma_1 \leftarrow \lp g^*_1, g_{1,1}, g_{1,2}, \alpha_1 \rp \in {\cal Q}_0$ for some $g_{1,1}, g_{1,2} \subset [n]$; 
		\itemc \breturn $t \lp \sigma_1 \rp$;
		\eenuc
	\itemb \bend \bif
	\eenub
\item \bend \bif
\eenu
\bend
\end{small}
\caption{To Construct a term $t \lp \sigma \rp$ at $\alpha$} 
\label{LocalShift}
\end{figure}

Step 3 of {\sc Shift} calls {\sc LocalShift} to construct a term $t \lp \sigma \rp$ for every $\sigma \in {\cal Q}_0$. It is based on $f \lp \sigma \rp$ output by {\sc BlockedEdges}.

Let us now show the independence property. It is:
\beeq{IndependenceProperty}
d \lp f \lp \sigma \rp \rp \cap {y_{j'} \choose 2}=\emptyset
\textrm{~~~for $\sigma \in {\cal Q}_j$ such that $j \ne j'$.}
\eeq
Because of the property, we can focus on the current $y_j$ to prove \refeq{BasicProperty}. The $k$-clique problem in the universal space $[n]$ has been partitioned into smaller problems in $y_1, y_2, \ldots, y_q$ this way; the existence of all the generated $k$-cliques $c$ is reduced to the independent existence of some $c$ in each $g \cup y_j$ at every $\alpha$. The construction is a consequence of \refth{ExtensionGeneratorTheorem}.

See the following claim.

\begin{lemma} \label{fSigma}
$d \lp f( \sigma) \rp \subset {g \cup y\lp \sigma \rp \choose 2} \setminus {g \choose 2}$ for every $\sigma \in {\cal Q}_0$. Thus, each edge in $d \lp f\lp \sigma \rp \rp$ has an endpoint in $y \lp \sigma \rp$. 
\end{lemma}
\begin{proof}
Fix $\sigma$ and let $\sigma^*=f \lp \sigma \rp=\lp g, g_1^*, g_2^*, \alpha \rp$. Step 1-2-1 of {\sc Shift} and Step 1-1 of {\sc BlockedEdges} together choose a $k$-clique $c^*$, such that  $\sigma^*$ is incident is $c^*$ at $\alpha$ and $g \subset c^* \subset g \cup y \lp \sigma \rp$. The sets $g^*_i$ are included in $c^*$ since $\sigma^*$ is incident to $c^*$. Thus, 
\[
d(f(\sigma))=d\lp \sigma^* \rp={g_1^* \cup g_2^* \choose 2} \setminus {g \choose 2}
\subset {g \cup y\lp \sigma \rp \choose 2} \setminus {g \choose 2}.
\]
The second statement follows the first. 
\end{proof}

An edge in ${y_{j'} \choose 2}$ have both endpoints in $y_{j'}$. The edge set $d \lp f \lp \sigma \rp \rp$ for any $\sigma$ such that $y\lp \sigma \rp=y_j$ has at least one endpoint in $y_j$ by the lemma. Thus $d \lp f(\sigma) \rp \cap {y_{j'} \choose 2}=\emptyset$ if $j \ne j'$, $i.e.$, the independence property \refeq{IndependenceProperty}.

We now show \refeq{BasicProperty}.

\medskip

\begin{lemma} \label{Phi}
$d \lp f \lp \sigma \rp \rp \cap z=\emptyset$ for every $\sigma \in {\cal Q}_0$. 
\end{lemma}
\begin{proof}
We claim $d \lp f(\sigma) \rp \cap z_j = \emptyset$ for each  $\sigma \in {\cal Q}_0$ and $j \in [q-1]$. Suppose to the contrary $d \lp f(\sigma) \rp \cap z_j \ne \emptyset$. The independence property means $y \lp \sigma \rp =y_j$ since $z_j \subset {y_j \choose 2}$. The choice $y \lp \sigma \rp=y_j$ is made by Step 1-1 of {\sc BlockedEdges} when $\sigma$ belongs to the current ${\cal Q}$. By Step 1-5, $\sigma \in {\cal Q}$ is stored in ${\cal Q}(z_j)$ due to $d \lp f(\sigma) \rp \cap z_j \ne \emptyset$. Step 1-6 stores it in the new ${\cal Q}$ so that $y \lp \sigma \rp \ne y_j$ by Step 1-1. A contradiction. The proves the claim $d \lp f(\sigma) \rp \cap z_j = \emptyset$. 

In addition, $d \lp f(\sigma) \rp \cap z_q = \emptyset$ for any  $\sigma \in {\cal Q}_0$ due to the independence property and ${\cal Q}_q=\emptyset$ in \refco{SizeOfQj}; there is no $\sigma$ such that $y \lp \sigma \rp=y_q$ and $z_q \subset {y_q \choose 2}$. 

Hence $d \lp f(\sigma) \rp \cap z_j = \emptyset$ for every  $\sigma \in {\cal Q}_0$ and $j \in [q]$, meaning \refeq{BasicProperty}.
\end{proof}

\medskip

We now show the disjointness property \refeq{DisjointnessProperty}. Let us see how {\sc LocalShift} constructs a term $t \lp \sigma \rp$ at $\alpha$ for $\sigma = \lp g, g_1, g_2, \alpha \rp$. Suppose $\alpha$ is a conjunction $\alpha_1 \wedge \alpha_2$. The argument is simpler when $\alpha$ is a disjunction or a leaf.

In the paragraph containing \refeq{BasicProperty}, we discussed that the function $\sigma^*=f(\sigma)$ represents the switch from a clique $c$ to $c^*$ such that $c^* \subset g \cup y(\sigma)$. Step 2-1 of {\sc LocalShift} retrieves this $\sigma^*=f(\sigma)$. Because $c^* \in \overline{{\cal C}_\alpha}$ (a non-error clique generated at $\alpha$ due to Step 1-2 of {\sc Shift}), and because of our construction by {\sc CliqueGenerators}, the $k$-clique $c^*$ is a non-error clique generated at both $\alpha_1$ and $\alpha_2$. There exists $\sigma_i \in {\cal Q}_0$ incident to $c^*$ at each $\alpha_i$ by \reflm{PropertyOfQ0} (ii). In other words, we have
\begin{quote}
{\em The Chain Property of $\sigma$:} Let $\sigma^*=f(\sigma) = \lp g, g_1^*, g_2^*, \alpha \rp$ for any $\sigma \in {\cal Q}_0$ at a conjunction $\alpha=\alpha_1 \wedge \alpha_2$. There exists $\sigma_i= \lp g_i^*, g_{i, 1}, g_{i, 2}, \alpha_i \rp \in {\cal Q}_0$ for each $i\in [2]$ and some $g_{i, 1}, g_{i,2} \subset [n]$. If $\alpha$ is a disjunction, there exists such $\sigma_i$ at one of $\alpha_i, i \in [2]$. 
\end{quote}

This allows for our construction of $t \lp \sigma \rp$ in {\sc LocalShift}; previously, $t \lp \sigma_i \rp$ has been found at $\alpha_i$. Step 2-2-2 returns $t \lp \sigma_1\rp \cup t \lp \sigma_2 \rp$ as $t \lp \sigma \rp$. Since $t \lp \sigma_i \rp \in DNF(\alpha_i)$, it is a term is at $\alpha$ by the definition of $DNF(\alpha)$. Also by induction hypothesis for \refeq{DisjointnessProperty}, $t \lp \sigma \rp= t \lp \sigma_1\rp \cup t \lp \sigma_2 \rp$ is disjoint from $z$ outside ${g^*_1 \cup g^*_2 \choose 2}$. By \reflm{Phi} and $d \lp f(\sigma) \rp={g^*_1 \cup g^*_2 \choose 2}\setminus {g \choose 2}$, $t \lp \sigma \rp \setminus {g \choose 2}$ is disjoint from $z$, showing \refeq{DisjointnessProperty}.

Find a formal proof of the disjointness property with the following lemma.

\begin{lemma} \label{GoodDominant} For each $\sigma \in {\cal Q}_0$, {\sc LocalShift} returns a term $t \lp \sigma \rp$ at $\alpha$ such that
\beeq{eqGoodDominant}
t \lp \sigma \rp \setminus {g \choose 2} \subseteq \Phi,
\eqwhere
\Phi \stackrel{def}{=}
\bigcup_
{\sigma  \in {\cal Q}_0} d \lp f(\sigma) \rp.
\eeq
\end{lemma}
\begin{proof}
Prove by induction on $depth(\alpha)$. Fix a given quadruple $\sigma=\lp g, g_1, g_2, \alpha \rp \in {\cal Q}_0$. The basis occurs when $\alpha$ is a leaf of $C$ associated with an edge $e \in {[n] \choose 2}$. Algorithm {\sc CliqueGenerators} collects all $c$ generated at $\alpha$. Every such $k$-clique contains $e$ in ${c \choose 2}$, so $g$ contains the two vertices incident to $e$. {\sc LocalShift} correctly returns $t \lp \sigma \rp= \lb e \rb \in DNF \lp \alpha \rp$ such that \refeq{eqGoodDominant}.

For induction step, we consider a conjunction $\alpha = \alpha_1 \wedge \alpha_2$ since a disjunctive case is shown similarly. Assume true for $\alpha_i$ and prove true for $\alpha$. By the chain property, there exists $\sigma_i$ at each $\alpha_i, i \in [2]$ for $\sigma^*=f\lp \sigma \rp$ found by Step 2-1 of {\sc LocalShift}. By induction hypothesis, {\sc LocalShift} has constructed a term $t \lp \sigma_i \rp$ at $\alpha_i$ for each $i=1,2$ and some $\sigma_i \in {\cal Q}_0$ at $\alpha_i$. Step 2-2-2 returns $t\lp \sigma_1 \rp \cup t \lp \sigma_2 \rp$ as $t \lp \sigma \rp$, so it is a term at $\alpha$. 

To verify  \refeq{eqGoodDominant}, let $f\lp \sigma \rp=\sigma^*=\lp g, g_1^*, g_2^*, \alpha \rp$. Observe
\(
&&
t \lp \sigma \rp \setminus {g \choose 2} 
=
\bigcup_{i=1}^2 t \lp \sigma_i \rp  \setminus {g \choose 2}
\nexteqline
\bigcup_{i=1}^2 \lp t \lp \sigma_i \rp  \cap {g^*_i \choose 2}  \setminus {g \choose 2} \rp
\cup 
\bigcup_{i=1}^2 \lp t \lp \sigma_i \rp  \setminus {g^*_i \choose 2} \setminus {g \choose 2} \rp.
\)
Since $d \lp f(\sigma) \rp={g^*_1 \cup g^*_2 \choose 2} \setminus {g \choose 2}$, 
\(
&&
\bigcup_{i=1}^2 \lp t \lp \sigma_i \rp  \cap {g^*_i \choose 2}  \setminus {g \choose 2} \rp
\subseteq
{g^*_1 \choose 2} \cup {g^*_2 \choose 2} \setminus {g \choose 2} 
\subseteq d \lp f \lp \sigma \rp \rp \subseteq \Phi.
\)
By induction hypothesis, $\bigcup_{i=1}^2  t \lp \sigma_i \rp  \setminus {g^*_i \choose 2} \subseteq \Phi$. Thus $t \lp \sigma \rp \setminus {g \choose 2} \subseteq \Phi$, proving the lemma.
 \end{proof}

\medskip

\begin{corollary} The disjointness property 
\refeq{DisjointnessProperty} holds for every $\sigma = \lp g, g_1, g_2, \alpha \rp \in {\cal Q}_0$. 
\end{corollary}
\begin{proof}
$\Phi \cap z=\emptyset$ by \reflm{Phi}. Thus $t \lp \sigma \rp \setminus {g \choose 2} \cap z \subseteq \Phi \cap z=\emptyset$. 
\end{proof}

\medskip

Step 4 of {\sc Shift} returns any $t \lp \sigma \rp$ at $r \lp C \rp$ as $t \lp {\bf y }\rp$. The only generator found by {\sc CliqueGenerators} at the root is $g= \emptyset$. The disjointness property \refeq{DisjointnessProperty} implies \refeq{MonotoneGoal}: $t \lp {\bf y}\rp \cap z = \emptyset$. By Step 3 of {\sc BlockedEdges}, $z=z_1 \cup z_2 \cup \cdots \cup z_q$. Each $z_j$ is chosen so that its removal leaves no $n^{1/5}$-clique in the space $[y_j]$. As a result, the removal of $z$ annihilates the $n^{1/5}q$-cliques in the universal space $[n]$. The shift $t \lp {\bf y }\rp$, disjoint from $z$, thus includes no $k$-clique since $k>n^{1/5}q$. It is a contradiction against the fact that every global term of $C$ contains a $k$-clique. This completes the proof of \refprop{Main2}.

\section{Discussions} \label{Remarks}

The presented shift method constructs a counter example for $C$ to compute {\sc Clique}$_{n, k}$. It partitions the clique problem in $[n]$ into $y_1, y_2, \ldots, y_q$ in a way that each of $g \cup y_1, g \cup y_2, \ldots, g \cup y_q$ contains a $k$-clique at every $\alpha$. As a result, we can independently find a blocked edge set $z_j$ in $y_j$ to integrate them for the global impossibility of generating all the $k$-cliques. This provides a new, set theoretical perspective of the hardness of monotone computation of cliques.

In this section, we compare the shift method with the standard proof based on the method of approximations, also discussing its possibility to apply it to a non-monotone circuit.

\subsection{Summary of the Standard Proof}

Let us review the standard proof of \refprop{Main2} \cite{papa} in the same setting as \refsec{MonotoneCircuitComplexity}. {\em Negative examples} are special truth assignments $S$ to which $C$ returns false. They are constructed by random coloring of vertices with $k-1$ colors; an edge $(v_1, v_2)$ exists if and only if the colors of $v_i$ are distinct. Clearly, such $S$ contains no $k$-cliques due to $k-1$ colors only. At a node $\alpha$, take $q$-sunflowers of generated $k$-cliques similarly to {\sc CliqueGenerators} finding extension generators. Replace them by the sunflower cores.  Continue until we no longer find such a $q$-sunflower. The final vertex sizes of the terms are $O \lp n^{\ep} \rp$ as shown below.

A {\em false positive $S$ at $\alpha$} is a negative example to which the node $\alpha$ originally returns false but the modified $\alpha$ returns true. A false positive may be newly created when we reduce terms at $\alpha$ with the detected sunflowers. Argue that this occurs only if every petal of a sunflower contains an edge missing in $S$. Its probability is sufficiently small, so false positives are maintained as a minority among the $\lp k-1 \rp^n$ negative examples throughout the process.

At the root of $C$, small terms of vertex size $O(n^\ep)$ cannot say no to the majority of remaining negative examples but the false positives: We find $q$-sunflowers of $vertex(t)$ where  $q=n^{5 \ep}$ and $t$ are terms at $\alpha$. (Here $vertex(\cdot)$ is the set of vertices incident to the argument edge set.) The number of false positives has been inductively upper-bounded by $(k-1)^{n}$ times $e^{-  \Omega \lp q \ln n \rp}$ by the above randomness argument. However, we have $t$ at the root of $C$ such that $|vertex(t)|=O\lp n^\ep \rp$. It falsely returns true to at least $(k-1)^{n -|vertex(t)|}=(k-1)^n e^{- O \lp n^\ep \ln n \rp}$ negative examples. This contradiction proves the impossibility of such $C$.

Lastly, we verify $|vertex(t)|=O\lp n^\ep \rp$ after the iterative application of the sunflower lemma. 

\begin{lemma}
With the current terms $t$ at $\alpha$, let $U$ be the family of $vertex(t)$ such that $|vertex(t)|=m \gg n^\ep$. If $U$ generates a set of $k$-cliques (family of $k$-sets) with sparsity at most $n^\ep$, then $U$ contains a $q$-sunflower.
\end{lemma}
\begin{proof}
By \refeq{SparsityBoundForUGeneratingV} and \reflm{asymptotic}, $\kappa \lp U \rp$ is at most 
\[
\ln {k \choose m} + n^\ep = \ln {n^{1/4} \choose m} + n^{\ep}
< m  \ln \frac{n^{1/4}}{m} + m + n^\ep.
\]
It means that the number of distinct $vertex(t)$ is at least 
$
{n \choose m} e^{-m  \ln \frac{n^{1/4}}{m} - m - n^\ep}.
$
It is straightforward to verify that this number is larger than the sunflower bound $(q-1)^m m! = e^{m \ln mq - O(m)}$ given by \reflm{ErdosRado}. Thus, such a family of $vertex(t)$ contains a $q$-sunflower. 
\end{proof}

\noindent
We keep taking sunflowers of $vertex(t)$ until it is no longer possible. Hence $|vertex(t)|=O \lp n^\ep \rp$ at the end of the process at each node $\alpha$.

\subsection{Toward Generalization of the Shift Method}

The heart of the shift method is the dynamic construction of the global term $t \lp {\bf y}\rp$ that is disjoint from the blocked edge set $z$. Here $z$ is chosen so that if it is eliminated, there exists no cliques of size $n^{1/5}q$. In other words, the algorithm {\sc Shift} finds $z$ that can be {\em emptied} in $t \lp {\bf y}\rp$.

This gives some hope to obtain a super-polynomial size lower bound of a non-monotone circuit $C$ computing cliques. For example, the following strategy may be possible: Given $C$, consider each edge set $z' \subset {[n] \choose 2}$ of a specified size such as $n^{11/6}q$. For every $k$-clique $c$, let $S$ be the truth assignment such that the edges in $z' \setminus {c \choose 2}$ are assigned negatively, and those in ${[n] \choose 2} \setminus z' \cup {c \choose 2}$ positively. Represent this $S$ by the pair $(c, z')$.

For each fixed $z'$, run {\sc Shift} on pairs $(c, z')$: Collect such $c$ at each $\alpha$ and find generators as before. {\sc LocalShift} constructs a global term $t \lp {\bf y }\rp$ such that there is no positive literal in $t \lp {\bf y }\rp \cap z \cap z'$ by the arguments for \refeq{DisjointnessProperty}. In other words, the edge set $z \cap z'$ is emptied in $t \lp {\bf y }\rp$. Due to possible negative literals in $C$, the global term may contain a contradiction, $i.e.$, $t \lp {\bf y }\rp \supset \lb e, \neg e \rb$ for an edge $e$.

If we can somehow find $z'$ such that {\sc Shift} empties $z \cap z'$, and such that $z=z'$, then we have the impossibility of a non-monotone circuit $C$ to compute cliques: There is no positive literals in $z \cap z' = z$ since {\sc Shift} empties it. There is no negative literal outside $z$, since every considered clique $c$ comes from the pair $(c, z')$ such that $z'=z$. Thus, the construct shift would include no contradiction. It is also free from a $k$-clique since the removal of $z$ annihilates smaller cliques. This could imply the impossibility of $C$.

Such $z'$ could be chosen with a majority of $l$-sets in the Hamming space as {\sc BlockedEdge} chooses $z_j$ in ${y_j \choose 2}$. If so, $z'$ is essentially obtained by a counting argument rather than inductive construction along the node structure of $C$. Such a method is probably not a natural proof. 

\medskip

In addition to a chance of applicability to a non-monotone circuit, it may be possible to improve the Alon-Boppana bound $e^{\Omega \lp \lp n / \log n \rp^{1/3} \rp}$ if the theory on the Hamming space is further developed. With a better extension generator, we could find a stronger sunflower bound to help the improvement.

\section{Conclusions}

We developed a theory on set theoretical properties of the Hamming space, with the emphasis on the $l$-extension of a family $U$ of $m$-sets. The extension generator theorem shows the existence of a small set $g$ to generate almost all the $\lp l - |g| \rp$-sets in $[n] \setminus g$. We demonstrated its applicability to the monotone complexity bound of {\sc Clique}$_{n, k}$: The theorem is used to partition the $k$-clique problem in a way that each of $y_1, y_2, \ldots, y_q$ contains a $k$-clique generated at any node $\alpha$, if joined with $g$. As a result, we can construct a global term called shift containing no $k$-clique for a monotone circuit $C$ of sub-exponential size. We also showed the existence of a sunflower with small core, which is not easy to construct by the sunflower lemma alone.

In addition, we discussed the possibility to apply the shift method to a non-monotone circuit to compute cliques, and to improve the Alon-Boppana bound with a better extension generator.

\medskip

\begin{acks}
The author would like to thank Professors Martin F\"urer at Penn State, Osamu Watanabe at Tokyo Institute of Technology, and Eric Allender at Rutgers University for their suggestions. Also, it has been very helpful of Professors Tadao Saito at the University of Tokyo and Hisashi Kobayashi at Princeton University to provide supports toward the submission of this paper. 
\end{acks}

\medskip\medskip

\appendix
\section*{Appendix: Proofs of the Claims in \refsec{BinomialAsymptotics}}
\setcounter{section}{1}

\medskip \medskip

\begin{small}

\quad {\sc Lemma} \ref{asymptotic}
\(
&&
\left| \ln {p \choose q} - q \lp \ln \frac{p}{q} + 1 - S \lp \frac{q}{p} \rp  \rp
- \frac{1}{2} \ln \frac{p}{2 \pi q(p-q)}
\right| = O \lp \frac{1}{\min \lp q, p-q \rp} \rp,
\)
for $p, q \in \Z$ such that $0<q < p$. 
\begin{proof}
By Stirling's series, $k! = \sqrt{2\pi k} \lp \frac{k}{e} \rp^k \lp 1+ O \lp 1 \big/ k \rp  \rp$ for $k \in \Z_{>0}$. It implies 
\(
{p \choose q} &=& \frac{p!}{q!(p-q)!}
=\frac{\sqrt{2 \pi p} \lp \frac{p}{e} \rp^p \lp 1+ O \lp p^{-1} \rp \rp}
{\sqrt{2 \pi q} \lp \frac{q}{e} \rp^p \lp 1+ O \lp q^{-1} \rp \rp \cdot
\sqrt{2 \pi (p-q)} \lp \frac{p-q}{e} \rp^{p-q} \lp 1+ O \lp (p-q)^{-1} \rp \rp}
\nexteqline
\sqrt{\frac{p}{2 \pi q(p-q)}}  \cdot \frac{p^p}{q^q (p-q)^{p-q}} \cdot r,
\) 
where $r=\frac{1+O\lp p^{-1}\rp}{\lp 1+O\lp q^{-1}\rp \rp \lp 1+O\lp (p-q)^{-1}\rp \rp}$ is a real number such that $|r-1|=O \lp 1 \Big/ \min \lp q, p-q \rp \rp$. Thus it suffices to show
\[
\ln \frac{p^p}{q^q (p-q)^{p-q}}  = q \lp \ln \frac{p}{q} + 1 - S \lp \frac{q}{p} \rp  \rp.
\]

Its left hand side is equal to
\[
p \ln p - q \ln q -(p-q) \lp \ln p + \ln \lp 1 - \frac{q}{p} \rp\rp = q \ln \frac{p}{q} - p \lp  1 - \frac{q}{p} \rp \ln \lp 1- \frac{q}{p} \rp.
\]
The function $S$ is defined by \refeq{DefOfSx}. With the Taylor series of the natural logarithm, we have
\( 
- \lp 1- x \rp \ln \lp 1 - x \rp &=& (1-x) \sum_{j \ge 1} \frac{x^j}{j}
=\sum_{j \ge 1} \frac{x^j}{j} -  \sum_{j \ge 1} \frac{x^{j+1}}{j}
=x + \sum_{j \ge 2} \frac{x^j}{j} -  \sum_{j \ge 1} \frac{x^{j+1}}{j}
\nexteqline
x + \sum_{j \ge 1} \lp \frac{x^{j+1}}{j+1} - \frac{x^{j+1}}{j} \rp
=
x - \sum_{j \ge 1} \frac{x^{j+1}}{j(j+1)}
\nexteqline
x - x \sum_{j \ge 1} \frac{x^{j}}{j(j+1)}= x - x \cdot S \lp x \rp,
\)
for $x \in (0,1)$. Therefore,
\(
\ln \frac{p^p}{q^q (p-q)^{p-q}} &=&
q \ln \frac{p}{q} - p \lp  1 - \frac{q}{p} \rp \ln \lp 1- \frac{q}{p} \rp
=q \ln \frac{p}{q}  +  p \lp \frac{q}{p} - \frac{q}{p}  \cdot S \lp  \frac{q}{p} \rp \rp
\nexteqline
q \lp \ln \frac{p}{q} + 1 - S \lp \frac{q}{p} \rp  \rp,
\)
proving the claim.
 \end{proof}

\medskip

\quad{\sc Lemma} \ref{basic2}
\[
{n-m \choose l}={n \choose l} e^{-\frac{lm}{n} - o\lp \frac{lm}{n} \rp},
\]
for $m, l \in [n]$ such that $l+m \ll n$.
\begin{proof}
First we see $\ln {n-m \choose l} <\ln {n \choose l} -\frac{lm}{n}$. Its proof is found in \cite{my6} as follows:
\(
&&
\frac{{n-m \choose l}}{{n \choose l}} 
=
\prod_{i=0}^{l-1} \frac{n-m-i}{n-i}
=
\prod_{i=0}^{l-1} \lp 1 - \frac{m}{n-i} \rp
< \lp 1 - \frac{m}{n} \rp^l,
\sothat
\ln {n-m \choose l}  - \ln {n \choose l}  \le l \ln \lp  1- \frac{m}{n} \rp
{} \\  && {}
= - l \sum_{i \ge 1} \frac{1}{i}  \lp \frac{m}{n} \rp^i
< - \frac{lm}{n}.
\)

The lower bound $\ln {n-m \choose l} \ge \ln {n \choose l} -\frac{lm}{n} - o \lp \frac{lm}{n} \rp$ is shown similarly:
\(
&& 
\frac{{n-m \choose l}}{{n \choose l}}  \ge 
\lp 1 - \frac{m}{n-l} \rp^l,
\sothat
\ln {n-m \choose l}  - \ln {n \choose l}  \ge l \ln \lp  1- \frac{m}{n-l} \rp
= - l \sum_{i \ge 1} \frac{1}{i}  \lp \frac{m}{n-l} \rp^i.
\)
We have $n \gg \max \lp l, m \rp$ and $l \sum_{i \ge 1} \frac{1}{i}  \lp \frac{m}{n} \rp^i= l \lp \frac{m}{n} + \frac{1}{2} \lp \frac{m}{n} \rp^2 + \cdots  \rp= \frac{ml}{n} + o \lp \frac{ml}{n} \rp$. Thus,
\[
\ln {n-m \choose l}  - \ln {n \choose l}  \ge - l \sum_{i \ge 1} \frac{1}{i}  \lp \frac{m}{n-l} \rp^i 
=
-\frac{ml}{n} - o \lp \frac{ml}{n} \rp, 
\]
proving the lemma.
\end{proof}

\medskip

\quad{\sc Lemma} \ref{basic3}
\[
\ln {l-m \choose m-j}{m \choose j}
\le
\ln {l \choose m}  - j \ln \frac{j l}{m^2} + j  + \ln j+ O(1),
\]
for $l, m \in [n]$ and $j \in [m]$ such that $m^2 \le l$. 
\begin{proof}
Assume $j<m$. (The case $j=m$ is proven in the footnote below\footnote{When $j=m$, LHS of the desired inequality is zero. Its RHS is
\[
R=\ln {l \choose m} - m \ln \frac{l}{m} + m + \ln m + \gamma.
\]
where $\gamma$ is a sufficiently large constant we may choose. \reflm{asymptotic} means $R=2m \pm o(m) + \gamma$ as $m \le \sqrt l \ll l$. The real number $R$ is positive when $m \gg1$. If $m=O(1)$, it is also positive with a choice of $\gamma$ since $|2m \pm o(m)|=O(1)$.}.)
Note that
\[
\ln {l-m \choose m-j} = \ln {l \choose m-j} - O \lp \frac{m^2}{l} \rp = \ln {l \choose m-j} - O(1),
\]
due to \reflm{basic2} and $l \ge m^2$. 

Approximate $\ln {l-m \choose m-j} = \ln {l \choose m-j} - O(1)$ and $\ln {m \choose j}$ by \reflm{asymptotic} as follows: For the former,
\(
&&
\ln {l-m \choose m-j} = X_1 + Y_1 \pm O(1),
\\ \textrm{where~~} &&
X_1=(m-j)\lp \ln \frac{l}{m-j} + 1 - S\lp \frac{m-j}{l} \rp \rp,
\\ \textrm{and~~} &&
Y_1=\frac{1}{2} \ln \frac{l}{(m-j)(l-m+j)}.
\)
For the latter,
\(
&&
\ln {m \choose j} = X_2 + Y_2 \pm O(1),
\\ \textrm{where~~} &&
X_2=j \lp \ln \frac{m}{j} +1 - S\lp \frac{j}{m} \rp \rp,
\\ \textrm{and~~} &&
Y_2=\frac{1}{2} \ln \frac{m}{j(m-j)}.
\)
Then it suffices to show:
\beeqn
&& \label{AppendixEq1}
X_1+ X_2 \le X_3 
- j \ln \frac{jl}{m^2} + j + O(1),
\\ \textrm{and~~} && \label{AppendixEq2}
Y_1 + Y_2 \le Y_3 + \ln j +  O(1),
\eeqn
where 
\(
&&
X_3=m \lp \ln \frac{l}{m} +1 - S \lp \frac{m}{l} \rp \rp,
\\ \textrm{and~~} &&
Y_3=\frac{1}{2} \ln \frac{l}{m(l-m)}.
\)

\medskip
First show \refeq{AppendixEq1}. We have
\(
&&
S\lp  \frac{m}{l} \rp = \sum_{k \ge 1} \frac{\lp \frac{m}{l} \rp^k}{k(k+1)}  = O \lp \frac{m}{l} \rp,
\sothat
S\lp  \frac{m-j}{l} \rp =S \lp \frac{m}{l} \rp - O \lp \frac{m}{l} \rp.
\)
by \refeq{DefOfSx}, $j < m$ and $m^2 \le l$. In addition,
\[
\ln \frac{l}{m-j} = \ln \frac{l}{m} - \ln \lp 1 - \frac{j}{m} \rp
= \ln \frac{l}{m} + \sum_{k \ge 1} \frac{1}{k} \lp \frac{j}{m} \rp^k,
\]
by the Taylor series of the natural logarithm. Thus, 
\(
X_1&=& (m-j)\lp \ln \frac{l}{m-j} + 1 - S\lp \frac{m-j}{l} \rp \rp
\nexteqline
\lp m- j \rp \lp \ln \frac{l}{m} +  \sum_{k \ge 1} \frac{1}{k} \lp \frac{j}{m} \rp^k
+ 1 - S\lp \frac{m}{l} \rp 
+ O \lp \frac{m}{l} \rp \rp
\nexteqline
m \lp \ln \frac{l}{m} + 1 - S \lp \frac{m}{l} \rp + O \lp \frac{m}{l} \rp \rp 
- j \lp \ln \frac{l}{m}  + 1 - S \lp \frac{m}{l} \rp + O \lp \frac{m}{l} \rp  \rp 
\\ && {} \spcd
+ (m-j)  \sum_{k \ge 1} \frac{1}{k} \lp \frac{j}{m} \rp^k
\\ &\le&
X_3 
- j \lp \ln \frac{l}{m}  + 1  \rp
+ (m-j)  \sum_{k \ge 1} \frac{1}{k} \lp \frac{j}{m} \rp^k
+  O \lp \frac{m^2}{l} \rp
\\ && {} \spcd
\lp \textrm{$X_3=m \lp \ln \frac{l}{m} +1 - S \lp \frac{m}{l} \rp \rp$ and $S\lp  \frac{m}{l} \rp = O \lp  \frac{m}{l} \rp$}\rp.
\)
Observe that
\(
(m-j)  \sum_{k \ge 1} \frac{1}{k} \lp \frac{j}{m} \rp^k
&=&
m  \cdot \frac{1}{1} \lp \frac{j}{m} \rp^1+ m \sum_{k \ge 2} \frac{1}{k} \lp \frac{j}{m} \rp^k
- j \sum_{k \ge 1} \frac{1}{k} \lp \frac{j}{m} \rp^k
\nexteqline
j + \sum_{k \ge 2} \frac{j^k}{km^{k-1}}
- \sum_{k \ge 1} \frac{j^{k+1}}{km^{k}}
\nexteqline
j  - \sum_{k \ge 1} \lp \frac{1}{k}  - \frac{1}{k+1} \rp \frac{j^{k+1}}{m^{k}} 
<j.
\)
By the above two,
\(
X_1&<& 
X_3  - j \ln \frac{l}{m}  + O \lp \frac{m^2}{l} \rp
=
X_3  - j \ln \frac{l}{m}  + O(1).
\)

On the other hand, 
\(
X_2&=&
j \lp \ln \frac{m}{j} +1 - S\lp \frac{j}{m} \rp \rp
< j \ln \frac{m}{j} + j.
\)
Therefore,
\[
X_1+ X_2 <  X_3 - j \ln \frac{l}{m}  + j \ln \frac{m}{j} + j + o(1)
=X_3 - j \ln \frac{jl}{m^2}  + j + O(1),
\]
proving \refeq{AppendixEq1}.

To show \refeq{AppendixEq2}, see that
\(
Y_1 + Y_2 - Y_3&=& 
\frac{1}{2} \ln \frac{l}{(m-j)(l-m+j)}
+ \frac{1}{2} \ln \frac{m}{j(m-j)}
-\frac{1}{2} \ln \frac{l}{m(l-m)}
\nexteqline
\ln \frac{m}{m-j}
- \frac{1}{2}\ln j
+ \frac{1}{2} \ln \frac{l-m}{l-m+j}
\\ &\le&
- \ln \lp 1- \frac{j}{m} \rp
- \frac{1}{2} \ln j.
\)
If $j \le m/2 $, the above is less than $- \ln \lp 1- \frac{2m}{m} \rp \le \ln 2 = O(1)$. Also, if $m$ is bounded by a constant, the maximum value of $- \ln \lp 1- \frac{j}{m} \rp$ is $- \ln \lp 1- \frac{m-1}{m} \rp =\ln m = O(1)$, since $j \le m-1$.

The remaining case occurs when both $j> m/2$ and $m \gg1$ are true. They mean
\[
Y_1 + Y_2 - Y_3 - \ln j < \ln m - \frac{3}{2} \ln j
< \ln m - \ln \lp \frac{m}{2} \rp^{3/2}<0.
\]
Hence, $Y_1 + Y_2 \le Y_3 + \ln j + O(1)$, proving \refeq{AppendixEq2}. This completes the proof. 
\end{proof}

\end{small}

\bibliographystyle{acmsmall}
\bibliography{acmsmall-sam}

\end{document}